\documentclass[12pt]{article}
\usepackage[affil-it]{authblk}
\usepackage{multirow}
\usepackage{comment}
\usepackage[normalem]{ulem}
\usepackage{titling}
\usepackage{subcaption}
\usepackage{epigraph}
\usepackage{graphicx}
\usepackage{color}
\usepackage{enumitem}
\usepackage{algorithm}
\usepackage[noend]{algpseudocode}
\usepackage{mdframed}
\usepackage{xr-hyper}
\usepackage[round]{natbib}
\usepackage{amsmath}
\usepackage{amssymb}
\usepackage{tikz}
\usetikzlibrary{arrows.meta,positioning}
\usepackage{amsthm}
\usepackage[centering,marginparwidth=0.9cm]{geometry}
\usepackage{etoolbox}
\usepackage{lipsum}
\usepackage{fullpage}
\usepackage[colorlinks,citecolor=northwestern-purple,urlcolor=chicago-maroon,linkcolor=chicago-maroon,backref=page]{hyperref}
\usepackage[nameinlink]{cleveref}
\usepackage{setspace}

\setlist[enumerate]{leftmargin=1.5cm,rightmargin=0.5cm,noitemsep, topsep=2pt}

\makeatletter
\newcommand*{\addFileDependency}[1]{
  \typeout{(#1)}
  \@addtofilelist{#1}
  \IfFileExists{#1}{}{\typeout{No file #1.}}
}
\makeatother

\definecolor{clemson-orange}{RGB}{234,106,32}
\definecolor{chicago-maroon}{RGB}{128,0,0}
\definecolor{northwestern-purple}{RGB}{82,0,99}
\definecolor{cornell-red}{RGB}{179,27,27}
\definecolor{sauder-green}{RGB}{171,180,0}
\definecolor{gray}{RGB}{192,192,192}
\definecolor{lawngreen}{RGB}{0,250,154}

\makeatletter
\def\BState{\State\hskip-\ALG@thistlm}
\makeatother
\allowdisplaybreaks

\newcommand{\E}{\bb E}

\theoremstyle{definition}
\newtheorem{fact}{Fact}
\newtheorem{theorem}{Theorem}
\newtheorem{lemma}{Lemma}
\newtheorem{corollary}{Corollary}
\newtheorem{proposition}{Proposition}

\newtheorem{definition}{Definition}

\newtheorem*{remark*}{Remark}

\theoremstyle{definition}

\makeatletter
\patchcmd{\@addmarginpar}{\ifodd\c@page}{\ifodd\c@page\@tempcnta\m@ne}{}{}
\makeatother
\crefname{assumption}{Assumption}{Assumptions}
\crefname{lemma}{Lemma}{Lemmas}
\crefname{theorem}{Theorem}{Theorems}
\crefname{corollary}{Corollary}{Corollaries}
\crefname{proposition}{Proposition}{Propositions}
\crefname{claim}{Claim}{Claims}
\crefname{procedure}{Procedure}{Procedures}
\crefname{algorithm}{Algorithm}{Algorithms}
\crefname{figure}{Figure}{Figures}
\crefname{remark}{Remark}{Remarks}
\crefname{section}{Section}{Sections}
\crefname{procedure}{Procedure}{Procedures}
\crefname{example}{Example}{Examples}
\crefname{definition}{Definition}{Definitions}
\crefname{table}{Table}{Tables}
\crefname{equation}{}{}
\crefname{enumi}{}{}
\crefname{conjecture}{Conjecture}{Conjectures}
\crefname{step}{Step}{Steps}

\def \b{\beta}

\def \T{\Theta}

\def \E{\mathbb{E}}

\def \ll{\lower1.6truept\hbox{${{\scriptstyle =\atop \scriptstyle <}}$}}
\def \gl{\lower1.6truept\hbox{${{\scriptstyle >\atop \scriptstyle
=}\atop{\scriptstyle <}}$}}
\def \lg{\lower1.6truept\hbox{${{\scriptstyle <\atop \scriptstyle =}\atop
{\scriptstyle >}}$}}

\thanksmarkseries{arabic}

\input{tcilatex}
\begin{document}

\title{\textbf{Top Trading Cycles in Large Markets: The Asymptotic Irrelevance of Priorities}%
\thanks{Earlier versions of this paper circulated under the title ``An Analysis of Top Trading Cycles in Two-sided Matching Markets.'' The present paper supersedes those earlier versions.}}
\author{
Yeon-Koo Che\thanks{%
Department of Economics, Columbia University, USA. Email: \href{mailto: yeonkooche@gmail.com}%
{\texttt{yeonkooche@gmail.com}}.} \, \and \, Olivier Tercieux\thanks{%
Department of Economics, Paris School of Economics, France. \ Email: \href{mailto: tercieux@pse.ens.fr}%
{\texttt{tercieux@pse.ens.fr}}.}
}
\date{\today \endgraf }
\maketitle

\begin{abstract}
Top Trading Cycles (TTC) is Pareto efficient and strategy-proof and explicitly uses agents' priorities. Although TTC favors higher-priority agents in each round, we show that this priority advantage vanishes as the market grows large under a canonical random model of preferences and priorities. In the limit, TTC produces assignments with virtually the same incidence of justified envy as Random Serial Dictatorship (RSD)---a mechanism entirely blind to priorities. This stark asymptotic equivalence implies that TTC effectively fails to satisfy standard fairness criteria in large markets, casting significant doubt on its practical appeal for balancing efficiency and fairness.

\end{abstract}

\section{Introduction}

In resource allocation problems with indivisible objects, agents are assigned objects based on both their preferences and their priorities. This framework captures many critical real-world scenarios, including the assignment of students to public schools, tenants to public housing, and human organs to transplant patients.  Agents' priorities—typically represented by object-specific rankings—reflect merit or other policy goals, such as test scores and walk zones in school choice, or waiting time and medical urgency in deceased-donor organ allocation. While achieving a Pareto efficient allocation is a clear desideratum in any market, respecting these priorities is an equally important policy objective. Equivalently, mechanisms strive to eliminate justified envy—a situation in which an agent envies another despite having a higher priority for the assigned object (\cite{balinski/sonmez:99} and \cite%
{abdulkadiroglu/sonmez:03}).\footnote{%
Respecting priorities is also part of the stability requirement of \cite%
{gale/shapley:62}.}

When efficiency is the primary objective,\footnote{Efficiency is fundamentally incompatible with strictly respecting priorities (\cite{roth:82}). If respecting priorities is instead the primary objective, the agent-proposing Deferred Acceptance algorithm of Gale and Shapley achieves this goal with the minimal sacrifice to efficiency: it produces a non-wasteful assignment that respects priorities and Pareto dominates all other non-wasteful, priority-respecting assignments.} economists typically recommend the (prioritized) Top Trading Cycles (TTC) algorithm for assigning agents to objects. Following \cite{abdulkadiroglu/sonmez:03}, who formally introduced this mechanism to the school choice literature, TTC was either adopted or considered as a serious candidate in  school assignment reforms in districts such as  New Orleans and Boston.\footnote{TTC was recommended by a Student Assignment Task Force in Boston in 2005. It was also used until recently in the New Orleans school system to assign students to public high schools. A version of TTC is currently utilized in France for assigning teachers to primary schools (see \cite{cdttu:22}). Furthermore, generalizations of TTC have been applied to kidney exchange among donor-patient pairs with incompatible donors (see \citet{Sonmez-Unver:11}).}


The TTC algorithm operates in successive rounds. In each round, agents point to their most preferred remaining objects, while objects point to the highest-priority remaining agents. Whenever a cycle forms, the agents associated with the cycle are assigned the objects they pointed to, and both are subsequently removed from the market.\footnote{A formal description is provided in   \Cref{sec:model}. The prioritized TTC mechanism studied here differs from the classical Shapley--Scarf TTC, whose primitives rely on initial ownership rights rather than institutional priorities (see \cite{shapley/scarf:74}).} While TTC is both Pareto efficient and strategy-proof, it is not the only mechanism satisfying these properties; Random Serial Dictatorship (RSD) is arguably its most prominent alternative. RSD randomly orders the agents and allows them to sequentially select their top choices among the remaining objects. Although both mechanisms guarantee efficiency, they diverge fundamentally in their treatment of priorities: TTC explicitly embeds them into the allocation process, whereas RSD ignores them entirely. By construction, TTC favors agents with high priorities. If an agent holds the highest priority for her most preferred remaining object at any stage of the procedure, TTC guarantees that she will receive it. This effectively eliminates justified envy for that object, as no lower-priority agent can subsequently claim it.\footnote{Any agent who envies her must still be present at that stage of the TTC algorithm, and therefore necessarily has a strictly lower priority for that object.} Exploiting this exact property, \cite{acprt:17} demonstrate that, in one-to-one matching environments, TTC minimizes justified envy among all strictly Pareto-efficient and strategy-proof mechanisms.\footnote{\cite{Dogan/Ehlers:22} extend the \cite{acprt:17} minimality result to a broader class of stability comparisons.}

However, this theoretical minimality leaves open a crucial quantitative question: to what extent does TTC \emph{actually} respect priorities? How substantial is TTC's practical advantage over priority-blind mechanisms---like RSD---in reducing justified envy? Although concerns that TTC might sacrifice priorities too heavily have historically limited its appeal in policy circles (e.g., \cite{pathak;16}), the literature has largely overlooked the magnitude of this fairness loss. A central contribution of this paper is to rigorously quantify the fairness cost associated with adopting this widely recommended efficient mechanism. Surprisingly, we show that in a canonical large-market environment, TTC's advantage in respecting priorities effectively disappears.


We consider a one-to-one matching market with $n$ agents and $n$ objects, where agents' preferences and priorities for objects are drawn uniformly at random. We show that as the market grows large ($n\to \infty$), the outcome of TTC, in terms of the joint distribution of agents' preference ranks \emph{and} their realized priority ranks, becomes asymptotically equivalent to the corresponding priority-blind benchmark.  This implies that, in large markets, TTC is virtually no better than RSD at respecting priorities. More precisely, the proportion of agents with justified envy (or those whose priorities are respected) under TTC becomes indistinguishable from that under RSD, both on average and in a probabilistic sense, as the market grows large.  

The reason for this striking result is the way TTC uses priorities to allocate objects. If an object is assigned via a \emph{short cycle} (a cycle of length two, where an agent points to an object and the object points back to her), it is impossible for an agent's priority for that object to be violated. However, if an object is assigned via a \emph{long cycle} (a cycle of length more than two), the acquiring agent has no inherent reason to have a high priority for that object. She acquires the object because she has a high priority for another object that she traded off. In fact, any envy towards an agent assigned via a long cycle is likely ``justified'' with a probability of no more than one-half, just as with RSD. Our irrelevance result stems from the fact that as the market grows ($n\to\infty$), the proportion of agents assigned via short cycles vanishes and their priority ranks are virtually uniform random just as under the RSD.

These results carry significant policy implications. Although economists frequently advocate adopting TTC in school choice settings, policymakers have historically hesitated to do so. This reluctance stems primarily from two concerns, vividly illustrated by the Boston Public Schools case (see \cite{abdulkadiroglu/pathak/roth/sonmez:06} and \cite{sonmez:23}). First, Boston officials viewed TTC's potential for widespread instances of justified envy as a major drawback. They were uncomfortable with students effectively ``trading their priorities," perceiving it as a source of substantial unfairness. Second, while TTC is theoretically strategy-proof, its mechanics may not be obvious or straightforward for participants to navigate.\footnote{%
\cite{li2017obviously} defines the notion of an extensive-form implementation of a strategy-proof mechanism that makes the strategy-proofness obvious to participants. It is well known that the RSD has an obviously strategy-proof (OSP) implementation (i.e., the sequential version of RSD), whereas TTC has no OSP implementation.} Although recent progress has made TTC more transparent,\footnote{%
\cite{leshno/lo:17}, \cite{gonczarowski2023strategyproofness}, and \cite{Katuscak:24} have shown how TTC can be presented in ways---either through appropriate cutoffs or through a ``budget'' set of choices---that participants find more straightforward to understand and more transparent to play using truthful-reporting strategies.} the severity of
the first concern—namely, the extent of the fairness and stability losses relative to RSD as
a worst-case benchmark—has remained unclear. This paper directly addresses that gap by quantifying TTC's fairness loss in a canonical random environment.  

While the asymptotic irrelevance of priorities may seem intuitive in hindsight, formally proving it requires a precise characterization of how TTC allocates objects within our random model. A crucial step in our analysis involves establishing a novel Markov property:   the
number of  objects assigned in each round of TTC follows a simple Markov chain, depending
solely on the number of agents and objects present at the beginning of that round.    This Markov characterization enables us to prove that TTC terminates in 
a number of rounds that is sublinear in market size. Furthermore, we establish that the expected number of objects assigned via short cycles never exceeds two per round. Together, these findings indicate that the proportion of objects assigned via short cycles approaches zero asymptotically, ultimately rendering priorities irrelevant in large markets.

We view this Markov characterization of TTC as our second main contribution—one that is independent of our central irrelevance result and potentially valuable for broader market design applications. This characterization is theoretically significant because it provides a precise understanding of the aggregate TTC process.  As we detail in our discussion of the related
literature, our Markov characterization strengthens the celebrated equivalence result between
RSD and TTC \citep{knuth:96, abdulkadiroglu/sonmez:98}. While this latter literature establishes equivalence for the agent-side assignment/rank distribution for any finite economy, our current irrelevance result shows that the equivalence extends to the \emph{joint} distribution of preference ranks enjoyed by the agents and their realized priorities, in a large economy.\footnote{%
For instance, the distribution of priority ranks (i.e., for an object, the rank of the assigned agent in that object’s priority ordering) has not been characterized before. We characterize this distribution which leads to the irrelevance result described above as well as to the asymptotic instability of TTC in the large-market environment of \cite{che/tercieux:15}.}

 The current paper relates to several strands of literature. First, it contributes to the literature studying the tradeoff between efficiency and the elimination of justified envy, particularly in the school choice context. This tradeoff was first recognized by \cite{roth:82}, and subsequently confirmed by \cite{abdulkadiroglu/sonmez:03} in the school choice context, and by \cite{abdulkadiroglu/pathak/roth:09} in the context of weak priorities. \cite{che/tercieux:15} argue that this tradeoff remains significant under standard mechanisms (such as DA or TTC) even in a large market if there is sufficient correlation in agents’ preferences. Furthermore, \cite{acprt:17} demonstrate the sense in which TTC minimizes justified envy (or maximizes priority-respecting) within the class of Pareto-efficient and strategy-proof mechanisms in one-to-one matching settings. Viewed alongside the irrelevance result of the current paper, their finding suggests that our asymptotic irrelevance might be driven fundamentally by the strict requirements of Pareto efficiency and strategy-proofness, rather than by a specific feature of TTC itself.\footnote{%
The envy-minimality result of \cite{acprt:17} does not rule out the possibility of there being a Pareto-efficient and strategy-proof mechanism which may or may not be envy-minimal but reduces the incidences of justified envy significantly relative to RSD. Whether such a mechanism exists remains an open question.}

Second, as noted earlier, our irrelevance result is closely related to the celebrated equivalence among well-known random allocations recognized by several authors (\cite{knuth:96}, \cite{abdulkadiroglu/sonmez:98}, \cite{pathak/sethuraman:11}, \cite{carroll:14}, \cite{bade:14}). These authors show that for any finite economy, a class of random allocations—including TTC and RSD—implement the same agent-side random assignment, and hence the same distribution of agents' preference ranks. However, these results apply only to the {\it marginal} distribution of the preference ranks enjoyed
by agents. They are silent on the {\it joint} distribution of agents’ preference ranks and their
realized priorities under alternative mechanisms. As illustrated in our example, this joint distribution is exactly what dictates the extent to which agents’ priorities are respected and the extent to which they justifiably envy others. In a finite market, TTC and RSD are distinctly not equivalent in this regard. Nevertheless, we show that this equivalence is remarkably restored as the market grows large. Our result can therefore be viewed as a strengthening of the classic equivalence theorem. We demonstrate that this equivalence extends to the \emph{joint} distribution of priority and preference ranks in large economies where, in addition to priorities being drawn uniformly at random as in the classical setting, preferences are also drawn uniformly at random.

Third, our Markov characterization of TTC is closely related to a similar Markov characterization of the Shapley-Scarf TTC derived by \cite{frieze/pittel;95}.\footnote{%
Another related work is \cite{leshno/lo:17}, which studies TTC in a large market but with a very different asymptotics where the number of object types is finite while there are a continuum of copies/seats for each object type and a continuum of agents with finite preference types. This distinction makes the analysis largely unrelated. Nevertheless, we discuss how our irrelevance result extends to this alternative asymptotics (see  \Cref{sec:Discussion}).} In the Shapley-Scarf economy, agents hold \emph{property rights} over objects—exactly one object for each agent—and are allowed to trade their rights along cycles in successive rounds. Despite the close resemblance, the two mechanisms are distinct. The associated (``pointing'') map from objects to agents is always bijective (one-to-one) in the Shapley-Scarf TTC but not in the prioritized TTC, where distinct objects may point to the same agent. This structural difference leads to entirely different probabilistic dynamics in the associated composite map within our random economy, necessitating a novel approach and different arguments. The concepts of random spanning forests and random composite maps, which to our knowledge have never been applied in economics, prove crucial in our analysis and offer a methodological tool potentially useful for other economic applications.

The remainder of the paper is organized as follows.  \Cref{sec:example} uses a simple example to illustrate the main irrelevance result.  \Cref{sec:model} introduces the formal model and preliminary analytical tools. \Cref{sec:Markov Property} provides the Markov characterization.  \Cref{sec:irrelevance} presents the asymptotic irrelevance result and its implications.   \Cref{sec:Discussion} discusses the robustness and limitations of our findings.  In particular, we suggest that the irrelevance result holds much
more generally beyond the uniform iid draws of preferences and priorities.  However, we also note that irrelevance does not extend to environments where the number of copies for each object type grows fast; in such cases, TTC performs significantly better than RSD in reducing justified envy, consistent with \cite{acprt:17}.   

\section{Example}
\label{sec:example}

To see precisely why TTC's ability to respect priorities vanishes in large markets, we must first understand the micro-mechanics of how it assigns objects. Specifically, this example illustrates how TTC successfully eliminates justified envy when objects are assigned via \emph{short cycles}, but fails to do so when assignments occur via \emph{long cycles}. 

Consider a simple (deterministic) market with three agents ($1, 2, 3$) and three objects ($a, b, c$). We represent a round of TTC using a bipartite directed graph: agents point to their most preferred objects, and objects point to their highest-priority agents.

Consider first a configuration that generates short cycles:

\begin{center}
\begin{minipage}{0.72\textwidth}
\centering


\begin{tikzpicture}[>=Stealth, node distance=1.6cm]
\node[font=\small\bfseries] at (0,1.3) {Agents};
\node[font=\small\bfseries] at (4.6,1.3) {Objects};

\node (1) at (0,0.5) {$1$};
\node (2) at (0,-0.8) {$2$};
\node (3) at (0,-2.1) {$3$};

\node (a) at (4.6,0.5) {$a$};
\node (b) at (4.6,-0.8) {$b$};
\node (c) at (4.6,-2.1) {$c$};

\draw[->, bend left=18] (1) to (a);
\draw[->, bend left=18] (a) to (1);

\draw[->, bend left=18] (2) to (b);
\draw[->, bend left=18] (b) to (2);

\draw[->] (3) -- (a);
\draw[->] (c) -- (3);
\end{tikzpicture}
\end{minipage}
\end{center}

In this configuration, agents $1$ and $2$ point to their top choices ($a$ and $b$, respectively), and those objects point directly back to them. This mutual pointing forms the short cycles $(1,a)$ and $(2,b)$, which clear immediately. The critical observation here is that an object assigned via a short cycle strictly goes to the agent to whom it points---its highest-priority remaining claimant. Consequently, it is impossible for the assignment to violate priorities. No other remaining agent can have a stronger claim, completely eliminating the possibility of justified envy for that object.

Now consider a configuration that generates a long cycle:

\begin{center}
\begin{minipage}{0.72\textwidth}
\centering


\begin{tikzpicture}[>=Stealth, node distance=1.6cm]
\node[font=\small\bfseries] at (0,1.3) {Agents};
\node[font=\small\bfseries] at (4.6,1.3) {Objects};

\node (1) at (0,0.5) {$1$};
\node (2) at (0,-0.8) {$2$};
\node (3) at (0,-2.1) {$3$};

\node (a) at (4.6,0.5) {$a$};
\node (b) at (4.6,-0.8) {$b$};
\node (c) at (4.6,-2.1) {$c$};

\draw[->] (1) -- (b);
\draw[->] (b) -- (2);
\draw[->] (2) -- (a);
\draw[->] (a) -- (1);

\draw[->] (3) -- (a);
\draw[->] (c) -- (3);
\end{tikzpicture}
\end{minipage}
\end{center}

Here, TTC clears the long cycle $(1,b,2,a)$. Notice the mechanics of this trade: agent $1$ receives object $b$ not because she holds a high priority for it, but because she leverages her priority at $a$ to execute a trade. Similarly, agent $2$ acquires $a$ simply by participating in the cycle, bypassing the agent to whom $a$ actually points (agent $1$). Crucially, agent $2$'s priority for $a$ is completely irrelevant to her receiving it. Her priority could very well be lower than that of agent $3$, who also prefers $a$; in such a case, the assignment leaves agent $3$ justifiably envying agent $2$. In short, when objects are assigned via long cycles, priorities are effectively ignored or rendered irrelevant in a manner strikingly similar to RSD.

This contrast between short and long cycles underlies the observation by \cite{acprt:17} that TTC performs better than RSD at eliminating justified envy. TTC's theoretical advantage lies precisely in the assignment of objects via short cycles. The critical question, therefore, is how many agents and objects are actually assigned through short versus long cycles.

This is where the large-market asymptotics become central. Our main result demonstrates that as the market grows large, the fraction of objects assigned via short cycles completely vanishes. Because virtually all assignments are eventually made through long cycles, the mechanism's ability to respect priorities diminishes correspondingly.

To make this intuition more precise, consider a market with $n\geq 2$ agents and $n$ objects, where agents' preferences and priorities are drawn uniformly at random. In this environment, both RSD and TTC are Pareto efficient and induce the same probability distribution over assigned objects for each agent (\cite{pathak/sethuraman:11}; \cite{carroll:14}). However, when $n$ is small, TTC noticeably outperforms RSD in reducing justified envy, reflecting the significant role of short cycles in finite markets.

 \Cref{fig} illustrates how this difference evaporates as the market grows. Panel (a) shows the incidence ratio of justified envy over envy, or simply {\it the incidence ratio}.\footnote{The incidence ratio of justified envy over envy is calculated as the total number of individuals an agent justifiably envies, summed across all agents, and divided by the total number of envies (justified or not) experienced by all agents. This measure reflects the intensity of justified envy rather than simply the proportion of agents experiencing it. Alternative measures, such as the fraction of agents with justified envy or the fraction of blocking pairs, could be used, but they are less informative in this context.  For example, the fraction of blocking pairs converges to zero even for RSD when the number of objects equals the number of agents, while the agent-level fraction does not record how many justified-envy incidences a given agent has. Our chosen measure, therefore, captures the intensity of justified envy more directly; see \Cref{sec:Discussion}.  }  On the incidence-ratio metric in Panel (a), TTC's simulated performance converges to that of RSD.  As illustrated in  \Cref{fig}, TTC initially generates lower levels of justified envy and better realized priority ranks (i.e., for a given object, the assigned agent’s position in that object’s priority ordering). However, this advantage shrinks dramatically as the market grows, precisely because the short cycles through which TTC enforces priorities become increasingly rare.

Formally proving that short cycles vanish  requires considerable care. The difficulty is that the agents and objects remaining in later rounds of TTC are not a random sample of the original population. Rather, they are selected endogenously through the earlier realizations of the algorithm, meaning their preferences and priorities are conditioned on the history of play. This conditioning breaks the simple independence structure present at the outset and makes standard probabilistic arguments inapplicable. Overcoming this difficulty requires a more precise characterization of the stochastic process induced by TTC, which we develop in the next section via a Markov representation.

\begin{figure}[h]
\begin{tabular}{cc}
\includegraphics[width=8cm,height=8cm]{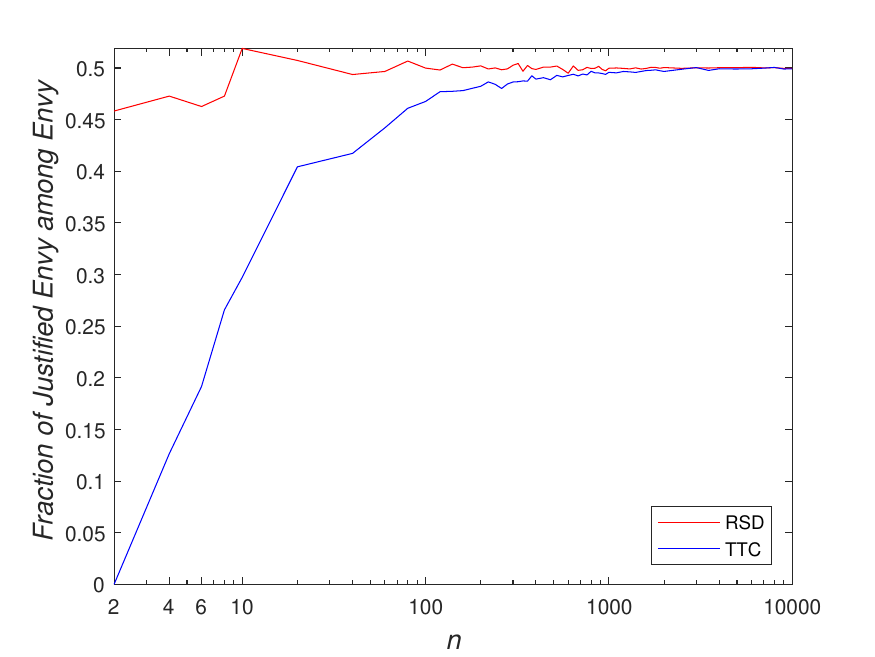} &
\includegraphics[width=8cm,height=8cm]{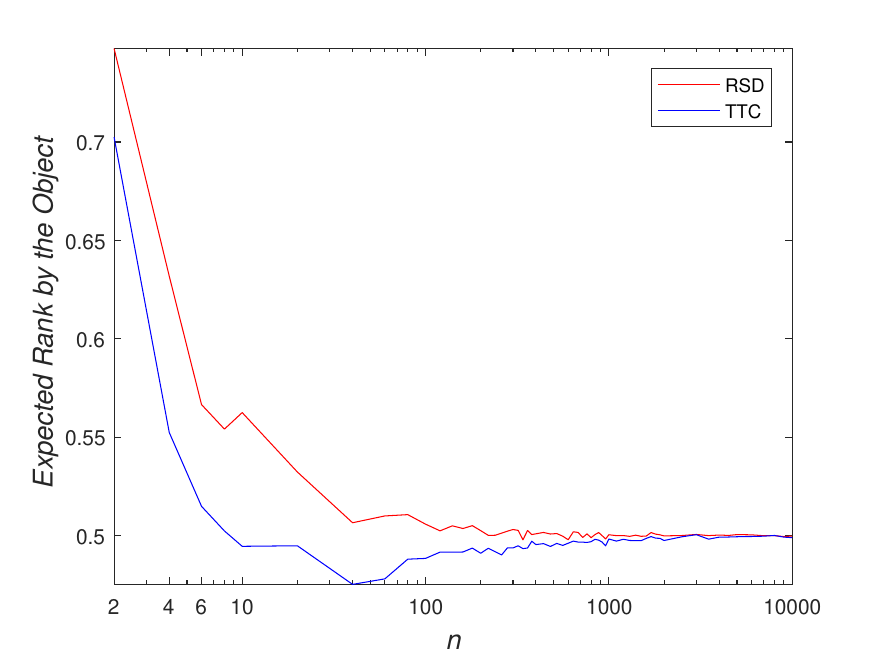} \tabularnewline
{\footnotesize (a) Incidences of justified envy among incidences of envy} &
{\footnotesize (b) Expected value of realized priority ranks}
\end{tabular}

\smallskip

{\footnotesize Note: The horizontal axis indicates the size of the market, i.e., the number of agents/objects. The figures plot averages over 100 (uniform iid) draws of agents' preferences and objects' priorities.}

\caption{Difference between TTC and RSD.}
\label{fig}
\end{figure}

\section{Model} \label{sec:model}

We consider a market with a set of agents ($I$) and a set of objects ($O$). The market may be unbalanced, meaning the number of agents and objects need not be equal. A \textbf{matching} is a map $\mu: I\to O\cup \{\o \}$ such that $i\ne i^{\prime }$ implies that either $\mu(i)\ne \mu(i^{\prime })$ or $\mu(i)=\mu(i^{\prime })={\o }$, where $\o $ denotes being unmatched. In  \Cref{sec:irrelevance}, we will focus on large markets, assuming an equal number of agents and objects ($|I|=|O|=n$), and analyze the outcomes as $n$ approaches infinity.

Each agent $i\in I$ has a strict preference ranking $P_i$ over the objects, and each object $o\in O$ has a strict priority ranking $\succ_o$ over the agents. We assume that all agents find all objects acceptable, and vice versa.\footnote{This is for convenience. Our result holds more generally if the number of acceptable partners grows linearly in $n$.} Let $(P,\succ):=(P_i, \succ_o)_{i\in I, o\in O}$ be a profile of agents' preferences and objects' priorities.

Throughout the paper, we consider a \textbf{random market} $(I,O,\tilde P, \tilde \succ)$ in which the profile $(\tilde P, \tilde \succ)$ is drawn \textit{iid} uniformly. Analyzing the limit properties of matching mechanisms under arbitrary preferences and priorities is generally intractable due to the curse of dimensionality; the number of possible preference and priority profiles increases exponentially as the market grows large. Consequently, it is standard practice in the large-market literature to analyze finite economies as empirical samples where types are drawn at random (see, e.g., \cite{wilson:72}, \cite{knuth/motwani/pittel:90}, \cite{pittel:92}, and \cite{ashlagi/kanoria/leshno:13}).\footnote{This uniform random type assumption enables the use of probabilistic methods that would otherwise be inapplicable.}

There are at least two ways to interpret this randomness in a real-world context such as school choice. The most natural interpretation is that a cohort of applicants each year draws their types (preferences and priorities) from this uniform iid distribution. Importantly, we do not view an actual cohort's preferences and priorities as being literally random; for any given cohort, they are fixed at a particular realization. Rather, the randomness serves as an analytical device to study the \textit{average} performance of the economy over many possible realizations. In particular, we focus on properties that hold with high probability in a sufficiently large market. Alternatively, the randomness can be viewed simply as a tool for evaluating the average performance of mechanisms—such as the average incidence of justified envy—across all possible profiles of student preferences and priorities. While the uniform iid assumption serves as an important canonical benchmark, we explain in  \Cref{sec:Discussion} how our results extend beyond this baseline.

In the case of RSD, there is additional randomness due to the mechanism's use of a random lottery. To accommodate this, we define a \textbf{state} $\omega$ as consisting of a profile $(P,\succ)$ of preferences and priorities, as well as a random serial order $\theta$ (a permutation of $I$). A \textbf{mechanism} is a mapping from each state to a matching. Given our random economy, randomness may arise from the types $(P, \succ)$ or the serial order $\theta$. Fix a realized state $\omega=(P, \succ, \theta)$. RSD and TTC are defined as follows:


\medskip  $\square$ \textbf{RSD:} For the realized serial order $\theta$ in the state $\omega$, at Step 1, the agent $\theta(1)$ selects her most preferred object according to $P_{\theta(1)}$ and exits the market. At each subsequent Step $t\geq 2$, agent $\theta(t)$ selects her most preferred object according to $P_{\theta(t)}$ from among those that remain. The procedure terminates once all individuals or all objects have been assigned. The RSD mechanism applies this procedure for every preference profile $(P,\succ)$ and every realization of the random order $\theta$. Note that RSD completely ignores the priority profile $\succ$.

\medskip
$\square$ \textbf{TTC:} In Round $t=1$, each individual $i \in I$ points to her most preferred object according to $P_i$. Each object $o \in O$ points to the individual who has the highest priority for that object according to $\succ_o$. Since the sets of individuals and objects are finite, the resulting directed graph contains at least one cycle. Every individual who belongs to a cycle is assigned the object she points to. All assigned individuals and objects are then removed. In Round $t\ge 2$, the same procedure is applied to the agents and objects remaining at the beginning of that round. The algorithm proceeds to Round $t+1$ unless all individuals or all objects have been assigned. This mechanism terminates in finite rounds because there is a finite set of agents, and at least one individual is removed at the end of each round. The TTC mechanism selects a matching via this algorithm for all possible profiles of agents' preferences and priorities. Note that TTC uses both $P$ and $\succ$ (and ignores $\theta$).

As noted, both the RSD and TTC mechanisms are \textbf{Pareto efficient}; that is, for each profile, the matching produced cannot be improved upon by an alternative matching that makes all individuals weakly better off and some strictly better off. Note that Pareto efficiency is defined solely with respect to agents' welfare. Furthermore, both mechanisms are \textbf{strategy-proof}, meaning it is a dominant strategy for each agent to report her preferences truthfully.\footnote{By reporting truthfully, each agent obtains a lottery over assignments that stochastically dominates any lottery she could obtain by misreporting her preferences.}

In fact, when priorities are drawn randomly, the two mechanisms are identical from the agents' perspective. For each preference profile $P$, the random allocations induced by the two mechanisms yield an identical lottery for each agent (\cite{pathak/sethuraman:11}; \cite{carroll:14}). However, this equivalence does not extend to the allocation from the perspective of the objects. To precisely compare the mechanisms along this dimension, we must understand the probabilistic structure of the allocation process in TTC.

\section{Markov Chain Property of TTC}
\label{sec:Markov Property}

Our first main result is a Markov chain characterization of TTC: the numbers of   agents and objects  remaining at the end of each round follow a simple Markov chain that depends only on
the numbers of agents and objects at the beginning of that round, according to well-defined
transition probabilities. The theorem is stated in a general setting, with the possibility of an imbalance between the number of agents and objects.

\begin{theorem}\label{Markov} Denoting $n_{t}$ and $o_{t}$ as the number of individuals and objects remaining in the market at any round $t$, the random sequence $\{(n_{1},o_{1}), (n_{2},o_{2}), ...\}$ is a Markov chain. If any round of TTC begins with $n$ agents and $o$ objects remaining in the market, the probability that there are $m\leq \min \{o,n\}$ agents assigned at the end of that round is\begin{equation*}p_{n,o;m}=\left( \frac{m}{(on)^{m+1}}\right) \left( \frac{n!}{(n-m)!}\right)\left( \frac{o!}{(o-m)!}\right) (o+n-m).\end{equation*}\end{theorem}

\begin{proof} See  \Cref{sec:Markov}.\end{proof}

The Markov property stated in  \Cref{Markov} is neither obvious nor anticipated. Because TTC operates in rounds, the algorithm’s progression in any given round depends highly sensitively on the specific circumstances of the history of all preceding rounds. In particular, although the agents' preferences and objects' priorities are uniformly random in the first round, the types of the agents remaining in subsequent rounds are definitively not. They remain in the market for endogenous reasons---namely, because they failed to form cycles.

This ``selection'' effect systematically skews the conditional distributions of the remaining preferences and priorities, breaking the initial independence structure.\footnote{%
To see this more precisely, suppose there are three individuals ($1, 2, 3$) and three objects ($o_1, o_2, o_3$). Per our assumption, the initial joint distribution of preferences and priorities is uniform iid. Now consider a possible history where agent 1 is assigned object $o_1$ in round 1. At the beginning of round 2, the joint distribution of the preferences and priorities of the \textit{remaining} agents ($2, 3$) with regard to the \textit{remaining} objects ($o_2, o_3$) is no longer uniform. Specifically, the probability that agent 2 prefers $o_2$ to $o_3$ is strictly less than $1/2$ \textit{conditional} on her having a higher priority than agent 3 at $o_2$; the simple intuition is that the mismatched pointing between $2$ and $o_2$ in round 1 is weighted heavily in the conditioning as a possible reason for their remaining in the market. The precise calculation is available from the authors.} 
This history dependence creates a conditioning issue that prevents us from invoking standard probabilistic arguments, such as the \emph{principle of deferred decisions}, in which one views each agent as \emph{drawing} preferences for the remaining objects at random in each round, rather than having drawn them at the beginning.

While tracking the precise conditional probabilities of any individual agent pointing to an object at the microscopic level is intractably complex,   \Cref{Markov} demonstrates that once we aggregate over all possible histories leading to a market with $n$ agents and $o$ objects, this microscopic complexity entirely washes out. The macroscopic state of the market---defined by the pair $(n,o)$ of remaining agents and objects---follows a memoryless Markov chain. The distribution of the number of agents assigned in any given round depends strictly on this current state, entirely independent of whether the market reached that state in Round 1 or Round 10.


\begin{figure}[t]\centering\includegraphics[width=0.95\textwidth]{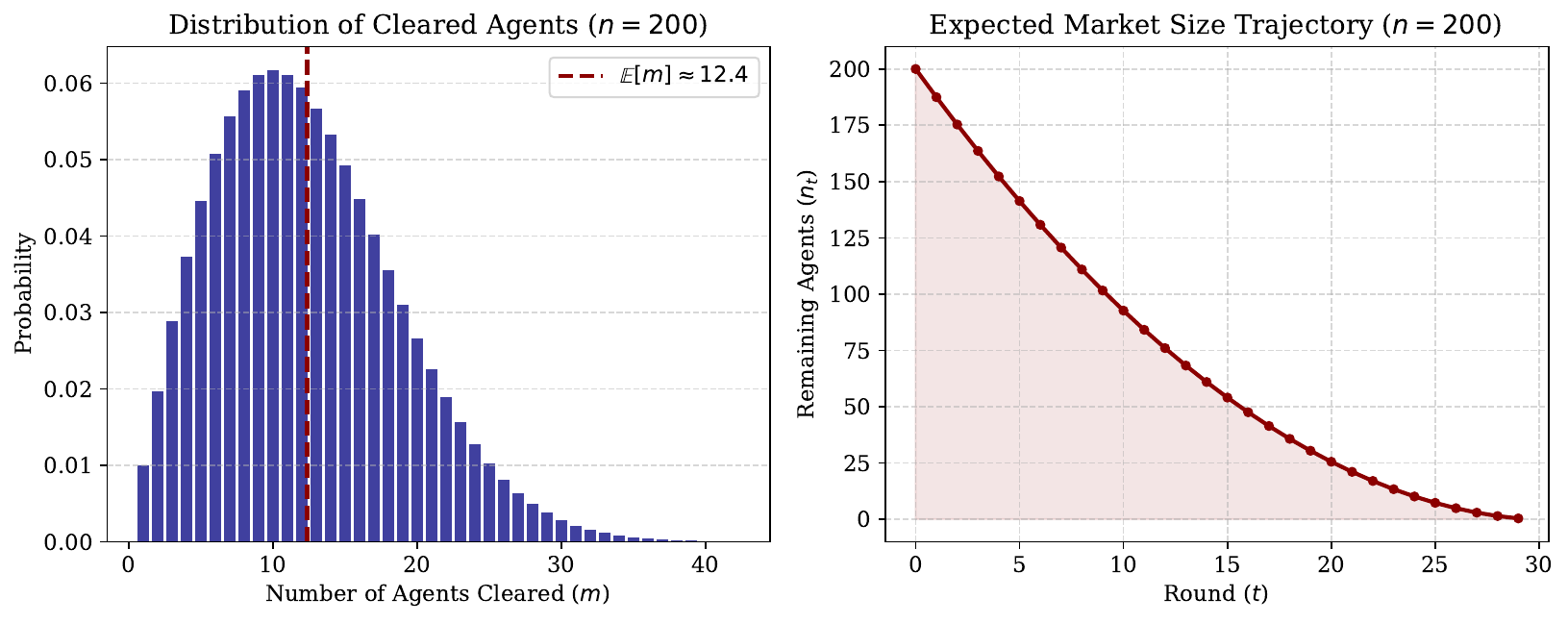}\caption{Markov dynamics of TTC in a balanced market ($n=o=200$).  The left panel shows the probability mass function for the number of agents cleared ($m$) in the first round; the distribution has a mean of order $\sqrt{n}$ with a proportional spread, as its standard deviation is also on the order of $\sqrt{n}$. The right panel illustrates the resulting expected trajectory of the market size over time, demonstrating a parabolic ``burn rate'' that clears the market in sublinear rounds.}\label{fig:markov_dynamics}\end{figure}

This aggregate predictability enables us to precisely visualize the algorithm's dynamics. As illustrated in   \Cref{fig:markov_dynamics}, if a round begins with $n$ agents and $n$ objects, the expected number of agents
cleared  is of order $\sqrt{n}$ (specifically, $\approx \sqrt{\pi n / 4}$ with a proportional standard deviation).  Because the ``burn rate'' of the market is proportional to the square root of its current size, the market collapses along a distinct parabolic trajectory.

To formally establish this Markov property, the crucial step in our analysis is to study the behavior of the so-called \emph{random spanning forests} induced by TTC. At the beginning of each round, when the directed edges to the agents and objects assigned in the previous round have been removed (but before the remaining participants have pointed to new targets), the remaining vertices and directed edges form a spanning rooted forest. The roots are the remaining agents and objects who had pointed in the previous round to parties that are now assigned; their outgoing edges have been removed, forcing them to randomly ``repoint'' to new partners. The other agents and objects continue pointing to the partners that still remain.  It is easy to see that the resulting sequence of random spanning forests is a Markov chain. The bulk of the proof consists of showing that, when this Markov chain is projected onto the numbers of remaining agents $(n)$ and objects $(o)$ at each round, it induces the simpler Markov chain in $(n,o)$ with the transition probabilities stated in  \Cref{Markov}.

Given this closed-form solution, we can derive the exact expectation and variance of the number of agents matched at each round (see Online \Cref{sec:number match at each stage}), enabling us to precisely visualize the algorithm's dynamics. As illustrated in  \Cref{fig:markov_dynamics}, the transition distribution has a nondegenerate $\sqrt n$-scale spread. If a round begins with $n$ agents and $n$ objects, the expected number of agents cleared is of order $\sqrt{n}$ (specifically, $\approx \sqrt{\pi n / 4}$). Because the ``burn rate'' of the market is proportional to the square root of its current size, the market collapses along a distinct parabolic trajectory. This characterization proves exceptionally useful for studying the \emph{completion time} of TTC---the number of rounds it takes for the algorithm to terminate. Driven by this $\sqrt{n}$ clearance rate, we can formally show that the completion time for TTC is strictly sublinear in $n$.\footnote{%
The argument is roughly as follows. The fact that an order $\sqrt{n}$ of agents are cleared on average starting from $n$ agents and objects means that with strictly positive probability (bounded away from zero), a very small number of objects (less than $\delta n$ for an arbitrarily small $\delta>0$) remain after an order of $\sqrt{n}$ rounds, as $n\to \infty$. This result further implies that with probability arbitrarily close to one, the completion time is less than $\alpha n$, for any arbitrarily small $\alpha>0$, as $n\to \infty.$}$^{,}$\footnote{%
It is known that the Shapley-Scarf TTC takes about $\sqrt{n}$ rounds (more precisely, $\sqrt{\frac{8}{\pi}n}- \frac{3}{\pi}\log(n)+O(1)$) to complete (see Theorem 1 of \citet{frieze/pittel;95}). In general, one would expect our prioritized TTC to take longer to complete than the Shapley-Scarf TTC. The reason is that, unlike our TTC, distinct objects always point to distinct agents in each round of the Shapley-Scarf TTC (due to the bijective ownership structure), leading to more cycles being formed in each round. Indeed, Online  \Cref{sec:Rounds} shows that the completion
time of prioritized TTC first-order stochastically dominates the completion
time of the Shapley--Scarf TTC. Nevertheless,  \Cref{prop:stopping} shows that, while slower, the completion time of prioritized TTC remains sublinear.}

\begin{proposition}\label{prop:stopping} Let $T$ denote the number of rounds required for TTC to conclude. Then, $\frac{T}{n}\overset{p}{\longrightarrow }0$.\end{proposition}
\begin{proof}See   \Cref{sec:stopping}.\end{proof}

The sublinear completion time will be critical for our purpose because, when combined with a bound on the expected number of short cycles per round, it will imply that short cycles become ``rare'' as the market grows large. 


Beyond this technical contribution, the Markov characterization serves several broader purposes. First, it enables a rigorous analysis of the distribution of ranks for individuals and objects under TTC, which is directly welfare-relevant.  Second, even though our baseline environment assumes uncorrelated preferences, \cite{che/tercieux:15} demonstrate that our
Markov result in \Cref{Markov} can be leveraged to make powerful predictions in richer environments where
agents’ preferences are positively correlated. 

\section{The Asymptotic Irrelevance of Priorities in TTC}
\label{sec:irrelevance}

We are now in a position to establish the asymptotic irrelevance of priorities under TTC. To this end, we consider a sequence of random markets with $n$ agents and $n$ objects. We demonstrate that, as $n \to \infty$, the \textbf{normalized priority ranks}---defined as the priority ranks divided by $n$---enjoyed by all objects under TTC converge to a uniform distribution over $[0,1]^n$. Crucially, this is the identical outcome produced by RSD, and this convergence occurs independently of the preference ranks enjoyed by the agents.

We begin with the observation that an object's assignment via a long cycle conveys no information about the priority of the acquiring agent. To be precise, fix an object $o$ and let $R_o$ denote the priority rank at $o$ of the agent \emph{who obtains $o$}, while $R_o^*$ denotes the priority rank of the agent \emph{to whom $o$ points} during the round it is assigned. We can show that $R_o$ is uniformly distributed over the set $\{R_o^*+1, \dots, n\}$, conditional on $R_o>R_o^*$ (i.e., conditional on $o$ being assigned via a long cycle).  This result stems from a symmetry argument: if an object $o$ is assigned via a long cycle, permuting the priority ranks at $o$ of all agents remaining in that round---excluding the agent $o$ points to, but including the agent who ultimately receives $o$---does not alter the outcome of TTC. Because every such permutation is equally likely, the assignment of object $o$ is uniformly distributed across the eligible remaining agents. Consequently, her priority rank is uniformly distributed over $\{R_o^*+1, \dots, n\}$. The details are provided in   \Cref{proof:corner stone}.

Given this observation, our irrelevance result naturally follows if we can establish two conditions: first, that TTC allocates virtually all objects via long cycles rather than short cycles; and second, that for virtually all objects $o$, the fraction $R_o^*/n \to 0$. Together, these ensure that the normalized priority rank $R_o/n$ converges to $U[0,1]$---exactly as in RSD---for almost all objects assigned via long cycles. 

Formally, let $\hat{O} := \{o \in O \mid R_o > R_o^*\}$ denote the set of all objects assigned via long cycles, and let 
\begin{equation*} 
\tilde{O} := \{o \in O \mid R_o^* \le \log^{1+\varepsilon}(n)\} 
\end{equation*} 
denote the set of objects that point to a top $\log^{1+\varepsilon}(n)$-ranked agent at the time of their assignment (where $\varepsilon > 0$ is an arbitrary strictly positive constant). We first prove that their intersection, $\bar{O} := \hat{O} \cap \tilde{O}$, eventually encompasses the entire population of objects as $n \to \infty$.

\begin{lemma}
\label{prop:short-cycles} $\frac{\left \vert \bar{O}\right \vert }{n}\overset{p}{ \longrightarrow }1$ as $n\to \infty$.
\end{lemma}
\begin{proof} It suffices to prove that $\frac{\left \vert \hat{O}\right \vert }{n}\overset{p}{\longrightarrow } 1$ and $ \frac{\left \vert \tilde{O}\right \vert }{n} \overset{p}{\longrightarrow } 1$.  We first establish the former. As demonstrated in  \Cref{prop: conditional exp short cycle} (in the Appendix), the expected number of objects matched via short cycles in any given round of TTC is  bounded above by $2$, irrespective of the algorithm's history. Let $S_t$ be the number of objects assigned through short cycles in round $t$, with $S_t=0$ if the algorithm has already terminated, and let $\mathcal F_{t-1}$ denote the history at the beginning of round $t$. Since $\{T\ge t\}$ is $\mathcal F_{t-1}$-measurable, \Cref{prop: conditional exp short cycle} gives
\[
\mathbb E[S_t\mathbf 1_{\{T\ge t\}}]
=\mathbb E\!\left[\mathbf 1_{\{T\ge t\}}\mathbb E[S_t\mid \mathcal F_{t-1}]\right]
\le 2\Pr\{T\ge t\}.
\]
Therefore
\[
\mathbb E\!\left[\sum_{t=1}^{T}S_t\right]
=\sum_{t\ge1}\mathbb E[S_t\mathbf 1_{\{T\ge t\}}]
\le 2\sum_{t\ge1}\Pr\{T\ge t\}=2\mathbb E[T].
\]
Since $T/n\to0$ in probability and $T\le n$, we have $\mathbb E[T]/n\to0$. Markov's inequality then implies that the fraction of short-cycle assignments is $o_p(1)$, i.e., $|\hat O|/n\to1$ in probability.

To establish that $\frac{\left \vert \tilde{O}\right \vert }{n}\overset{p}{\longrightarrow }1$, we introduce an auxiliary mechanism, TTC$^*$. This mechanism operates exactly like standard TTC, except that when cycles are cleared, the objects within each cycle are assigned to the agents that \emph{the objects point to} (rather than the agents pointing to them). Clearly, TTC$^*$ generates the exact same cycles in each round as TTC, removing identical sets of agents and objects. We then invoke Proposition 1 of \cite{che/tercieux:17}: in a balanced one-to-one random market with iid uniform preferences on both sides, any Pareto efficient matching assigns all but $o_p(n)$ agents to partners ranked among their first $\log^{1+\varepsilon} n$ choices, for every $\varepsilon>0$. We apply this result to the object side, treating each object's priority order as its preference order over agents. TTC$^*$ is Pareto efficient for the objects. Since TTC$^*$ clears the same cycles as TTC, the rank assigned to object $o$ under TTC$^*$ is exactly $R_o^*$. Hence $R_o^*\le \log^{1+\varepsilon}n$ for all but $o_p(n)$ objects.
\end{proof}

With these preliminary results in hand, we are now positioned to formally state our main theorem. In what follows, for any random variable $W$ taking values in $\{1,\dots,n\}$, we define its normalized counterpart as $\bar{W}:=\frac{1}{n}W$. Let $O^n:=\{o_{1},\dots,o_{n}\}$ and $I^n:=\{i_{1},\dots,i_{n}\}$ denote the sets of objects and individuals in an $n$-economy, respectively. Consider a sequence of collections of $2n$ random variables $\{V_{o_{1}},\dots,V_{o_{n}},V_{i_{1}},\dots,V_{i_{n}}\}$, each supported on $[0,1]$. Because the dimensionality of the assignment vector grows with the market size, we must rely on a specific notion of distributional convergence for growing random vectors:

\begin{definition}
A random vector $\{ \bar{W}_{o_{1}},\dots,\bar{W}_{o_{n}},\bar{W}_{i_{1}},\dots,\bar{W}_{i_{n}}\}$ \textbf{converges in distribution to} $\{V_{o_{1}},\dots,V_{o_{n}},V_{i_{1}},\dots,V_{i_{n}}\}$ as $n\to \infty$ if, for any integer $K$, any $x\in \lbrack 0,1]^{K}$, and any sequence $\{y^{n}\}$ taking values in $[0,1]^{n}$, we have
\begin{equation*}
\lim_{n\rightarrow \infty }\left \vert F^{n}(x,y^{n})-G^{n}(x,y^{n})\right \vert = 0,
\end{equation*}%
where $F^{n}$ is the cumulative distribution function (CDF) of $\{ \bar{W}_{o_{1}},\dots,\bar{W}_{o_{K}},\bar{W}_{i_{1}},\dots,\bar{W}_{i_{n}}\}$ and $G^{n}$ is the CDF of $\{V_{o_{1}},\dots,V_{o_{K}},V_{i_{1}},\dots,V_{i_{n}}\}$.
\end{definition}

From now on, let $\{\bar{R}_{o_{1}}, \dots, \bar{R}_{o_{n}}, \bar{R}_{i_{1}}, \dots, \bar{R}_{i_{n}}\}$ denote the profile of normalized priority and preference ranks under TTC. Under RSD, the serial order and agents' preferences determine the assignment, while object priorities are completely ignored. Thus, conditional on the RSD assignment, the realized priority ranks of the assigned agents are iid uniform on $\{1,\ldots,n\}$ and independent of the RSD preference-rank vector. The TTC-RSD equivalence result by \cite{pathak/sethuraman:11} implies that the rank distribution for agents is the same between the two mechanisms; the theorem below establishes that the conditional object priority ranks under TTC converge to an iid uniform distribution---the same as in RSD as $n\to \infty$.
 

Recall our earlier observation: any object assigned via a long cycle is uniformly allocated among the individuals whose priority rank at that object is strictly worse than $R^*_o$ (i.e., numerically larger than $R^*_o$). This implies that the rank realized by each object $o \in \bar{O}$ is stochastically dominated by the uniform distribution over $\{\lceil \log^{1+\varepsilon}(n) \rceil + 1, \dots, n\}$. Simultaneously, this realized rank trivially dominates the uniform distribution over $\{1, \dots, n\}$. As established in  \Cref{corner stone} (Appendix), these bounds hold independently of both the preference ranks enjoyed by the agents and the priority ranks enjoyed by the other objects within $\bar{O}$. Because the upper bound's truncation term vanishes in the limit ($\log^{1+\varepsilon}(n)/n \to 0$ as $n \to \infty$), this proposition, combined with  \Cref{prop:short-cycles} ensures that the joint distribution of these normalized ranks  converges to the iid uniform distribution as the market grows large in the sense stated in \Cref{thm: irrelevance}.

The main theorem now follows directly. 

\begin{theorem}[Irrelevance]
\label{thm: irrelevance}
(a) The profile $\{ \bar{R}_{o_{1}},\dots,\bar{R}_{o_{n}},\bar{R}_{i_{1}},\dots,\bar{R}_{i_{n}}\}$ of normalized priority and preference ranks under TTC converges in distribution to $\{ \bar{U}_{o_{1}},\dots,\bar{U}_{o_{n}},\bar{R}_{i_{1}},\dots,\bar{R}_{i_{n}}\}$ as $n\to \infty$.
(b) For any $x\in [0,1]$, $\frac{1}{n}\sum_{o\in O}\mathbf{1}_{\{ \bar{R}_{o}\leq x\}} \overset{p}{\longrightarrow} x$.
 \end{theorem}

\Cref{thm: irrelevance} is a rank-level statement.  Its fairness content is best seen by translating it into justified-envy statistics.  We use three normalizations.  The first is the incidence ratio plotted in \Cref{fig}: among all realized envy incidences, what fraction are justified?  The second counts agents rather than incidences: what fraction of agents experience at least one justified envy?  The third counts blocking pairs as a fraction of all nontrivial agent-object pairs.  These normalizations answer different questions, but they lead to the same conclusion: in the balanced iid one-to-one market, TTC and the priority-blind RSD become indistinguishable in both the incidence and the intensity of justified envy.  Notice that the agent-level statistic has a nondegenerate limit, whereas the blocking-pair fraction vanishes mechanically because a fixed agent-object pair is unlikely to be an envy pair at all.

\begin{proposition}[Justified-envy metrics]
\label{prop:je-metrics}
Consider the balanced iid one-to-one market.  The following statements hold.
\begin{enumerate}
\item The ratio of the expected number of justified-envy incidences to the expected number of envy incidences is exactly $1/2$ under RSD for every $n\ge2$, and converges to $1/2$ under TTC.
\item The expected fraction of agents with at least one justified envy converges to $2\log 2-1$ under both RSD and TTC.
\item Let $H_m:=\sum_{r=1}^m 1/r$.  The expected fraction of blocking pairs, normalized by $n(n-1)$, is
\[
        \frac{(n+1)(H_{n+1}-1)-n}{2n(n-1)}
        \sim \frac{\log n}{2n}
\]
under RSD, and is asymptotic to the same expression under TTC.  Hence, the blocking-pair fractions under both mechanisms vanish, but their ratio converges to one.
\end{enumerate}
\end{proposition}

\begin{proof}
See  \Cref{sec: BP}.
\end{proof}

The proposition compares aggregate fairness statistics.  The following corollary records the corresponding fixed-pair probabilities.  Unlike a priority-rank comparison that asks only whether agent $i$ has a higher priority than the recipient of object $o$, the corollary focuses on the economically relevant event that $i$ actually justifiably envies that recipient.

\begin{corollary}[Fixed-pair envy and justified envy]
\label{cor: BP}
Fix an agent-object pair $(i,o)$.  Under both RSD and TTC, the probability that $i$ envies the recipient of $o$ is
\[
        \frac{(n+1)(H_{n+1}-1)-n}{n^2}
        \sim \frac{\log n}{n}.
\]
Under RSD, the probability that this fixed pair is a blocking pair is exactly one-half of this probability.  Under TTC, it is asymptotically one-half of this probability.  Equivalently, conditional on $i$ envying the recipient of $o$, the probability that this envy is justified converges to $1/2$ under TTC, exactly as under RSD.
\end{corollary}

\begin{proof}
See \Cref{sec: BP}.
\end{proof}

\section{Discussion}
\label{sec:Discussion}

 The preceding irrelevance result suggests that TTC’s ability to use priorities to eliminate justified envy becomes increasingly limited as the market grows large.  While asymptotic irrelevance was established in an environment where preferences and priorities are drawn independently and uniformly at random, we report simulations suggesting that the same force is robust along several dimensions. We also discuss environments in which the priorities do not become irrelevant even in the limit.


\subsection{Correlated Preferences}

If agents' preferences are perfectly correlated, TTC completely eliminates justified envy. In this extreme scenario, all objects are trivially assigned via short cycles, yielding an assignment that is both efficient and strictly stable. By contrast, RSD would continue to generate a substantial amount of justified envy, even in the limit. However, perfect correlation is an extreme boundary case unlikely to materialize in real-world settings. 

A more empirically relevant case arises when preferences are only partially correlated. Suppose, for instance, that an agent's preferences are represented by a cardinal utility function $u_i(o) = u_o + \xi_{io}$, where $u_o$ captures a common utility derived from object $o$ that is shared by all agents, and $\xi_{io}$ represents an idiosyncratic, agent-specific utility drawn iid for each agent $i$. Provided the idiosyncratic component is nondegenerate, and its support is bounded, simulations suggest that the same irrelevance force survives in this richer environment.\footnote{To see this analytically, consider a simplified case where the support of the common utility is finite and its values are sufficiently spaced out such that objects are effectively ``tiered'': all agents prefer top-tier objects (those with the highest $u_o$), followed uniformly by second-tier objects, and so on. Under these conditions, the TTC mechanism effectively partitions into multiple sequential stages. In stage 1, all agents point to tier-1 objects; once these are fully assigned, stage 2 commences with the remaining agents pointing to tier-2 objects, and so forth. Because agents' priorities for each object remain i.i.d., each stage can be analyzed as an independent TTC market to which our asymptotic irrelevance result directly applies.} 
\begin{figure}[h]
\caption{Incidences of justified envy among incidences of envy.}
\label{fig2}
\begin{center}
\includegraphics[width=11cm,height=9cm]{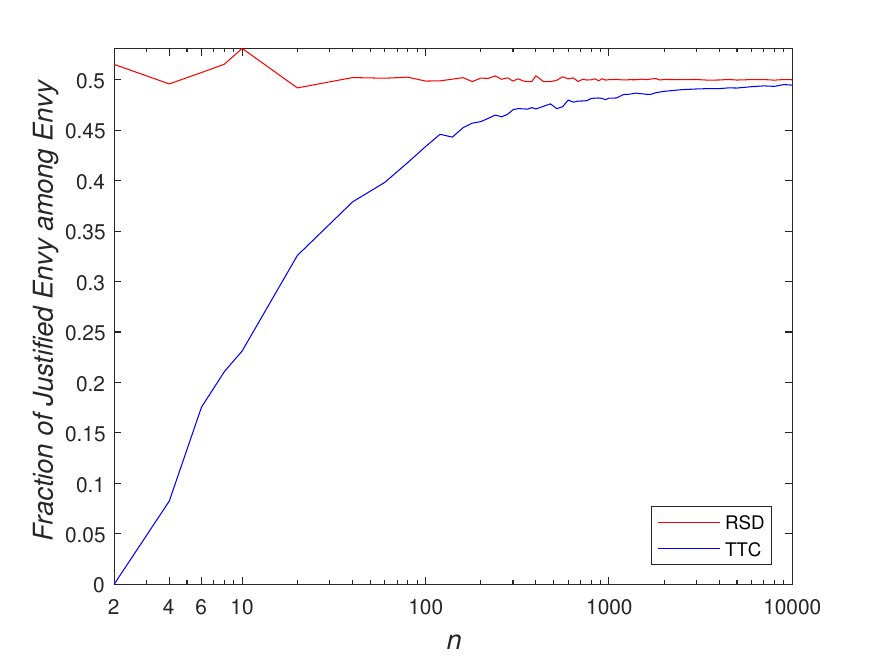}
\end{center}
\par
{\footnotesize {\ {Note: The horizontal axis indicates the number of
agents/objects. The figure plots the average proportions across 100 draws of
$(u_o, \xi_{io}, \eta_{io})$ from $U[0,1]^3$, where $\eta_{io}$ is the priority
score of agent $i$ at $o$.}} }
\end{figure}

As an illustration,   \Cref{fig2} shows that the simulation pattern extends to correlated preferences when $u_o$ and $\xi_{io}$ are each uniformly distributed over $[0,1]$.

\subsection{Correlated Priorities}

Suppose instead that agents' priorities are correlated. Such correlation naturally arises in settings like school choice, where students are often prioritized based on standardized test scores or overall academic performance. As in the case of correlated preferences, perfect correlation reduces TTC to a Serial Dictatorship governed by a suitably defined serial order.

To explore a more realistic, intermediate case, suppose that students are partitioned into a finite number of vertically ordered priority tiers—for example, 
broad test score bands. All objects strictly rank students in a higher tier above those in a lower tier, but within any given tier, students' relative priorities are drawn iid uniformly for each object.

In this tiered environment, the global asymptotic equivalence between TTC and RSD regarding justified envy breaks down. Because RSD relies on a single, uniform serial order drawn from the entire population, it frequently allows lower-tier students to select objects ahead of higher-tier students, generating widespread cross-tier justified envy. TTC, by contrast, naturally segregates the market by tiers through its pointing dynamics. Because all objects strictly prioritize higher-tier agents, they will exclusively point to top-tier agents as long as any remain. Consequently, these top-tier agents will almost surely be assigned via cycles consisting entirely of other agents within their same tier. This localized cycle formation cascades sequentially down the hierarchy, ensuring that justified envy is strictly limited to peers within the same tier and entirely eliminating cross-tier priority violations.

However, a striking form of our asymptotic irrelevance result persists. While TTC successfully protects the broad priority classes through its tier-by-tier clearing, the marginal distribution of an object's realized priority rank still converges to a uniform distribution over the relevant tier, exactly as it does under RSD. This within-tier equivalence occurs because the fraction of objects assigned via short cycles continues to approach zero as the market grows. In other words, while TTC respects the macro priority tiers, from the perspective of the objects, it effectively neutralizes the micro, within-tier priorities just as completely as a priority-blind mechanism.


\subsection{Many-to-One Matching}
\label{disc:many-to-one}

Our analysis thus far has focused on a one-to-one matching environment. While this serves as a canonical baseline, many real-world applications---such as school choice or public housing allocation---involve many-to-one matching, where multiple identical seats or units of an object type are available for assignment. 

Our simulations suggest that the irrelevance pattern extends to many-to-one matching environments  when the number of copies (seats) per object type grows sufficiently slowly relative to the number of object types.\footnote{This asymptotic framework---often referred to as the ``small school'' model---has been widely adopted in the literature (e.g., \cite{kojima/pathak:06}; \cite{ashlagi/kanoria/leshno:13}; \cite{che/tercieux:17}; \cite{che/tercieux:15}). It aligns well with empirical settings like the medical residency match (where roughly 20,000 doctors apply to 3,000-4,000 hospital programs) and the New York City public high school match (where the number of programs, roughly 800, far exceeds the average program capacity of around 100 students).} This extension is illustrated in  \Cref{fig3}, which plots the incidence of justified envy when the number of object types increases linearly with $n$, while the capacity of each object type is fixed at $20$.\footnote{The simulation
protocol for this exercise is described in Online Appendix~\ref{oa:simulation-details}.}   (Preferences and priorities remain drawn iid uniformly.)

\begin{figure}[h]
\caption{Incidences of justified envy among incidences of envy.}
\label{fig3}
\begin{center}
\includegraphics[width=11cm,height=9cm]{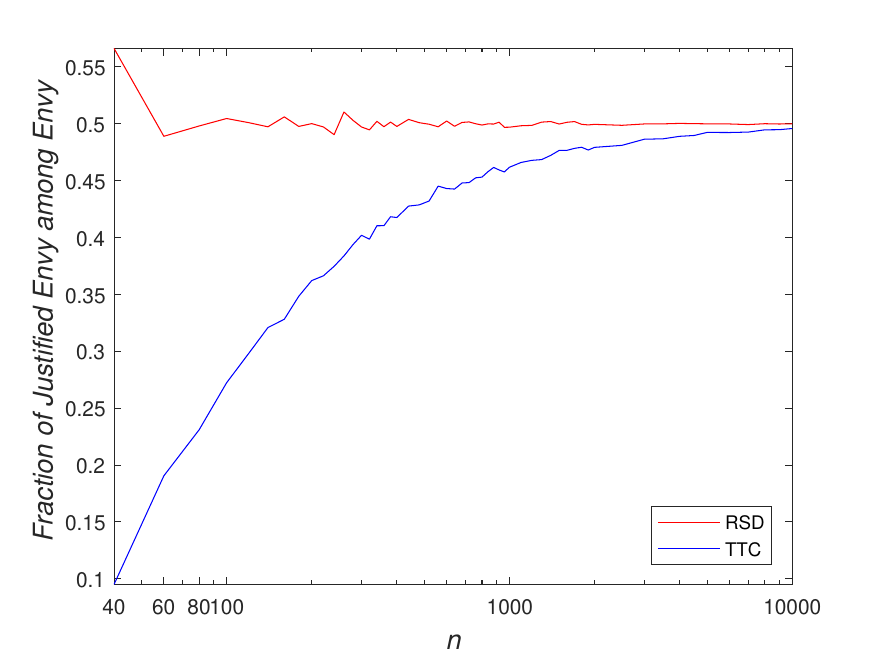}
\end{center}
\par
{\footnotesize {\ {Note: The horizontal axis indicates the number $n$ of
agents. The number of object types is equal to $n/20$. The figure plots averages
over 100 (uniform iid) draws of agents' preferences and objects' priorities.}%
} }
\end{figure}

However, our irrelevance result does not extend universally to the alternative asymptotic framework, where the capacity per object type grows rapidly relative to the number of object types.\footnote{This ``large school'' model is utilized by authors such as \cite{acy:15}, \cite{azevedo/leshno:16}, \cite{che/kim/kojima:13}, and \cite{leshno/lo:17}. It naturally captures school-choice environments in certain cities, where a handful of large schools each admit hundreds of students.} To see why, suppose the number of object types remains fixed at a finite number while the total number of agents and the capacity per object type grow large.

\begin{figure}[h!]
\begin{tabular}{cc}
\includegraphics[width=8cm,height=8cm]{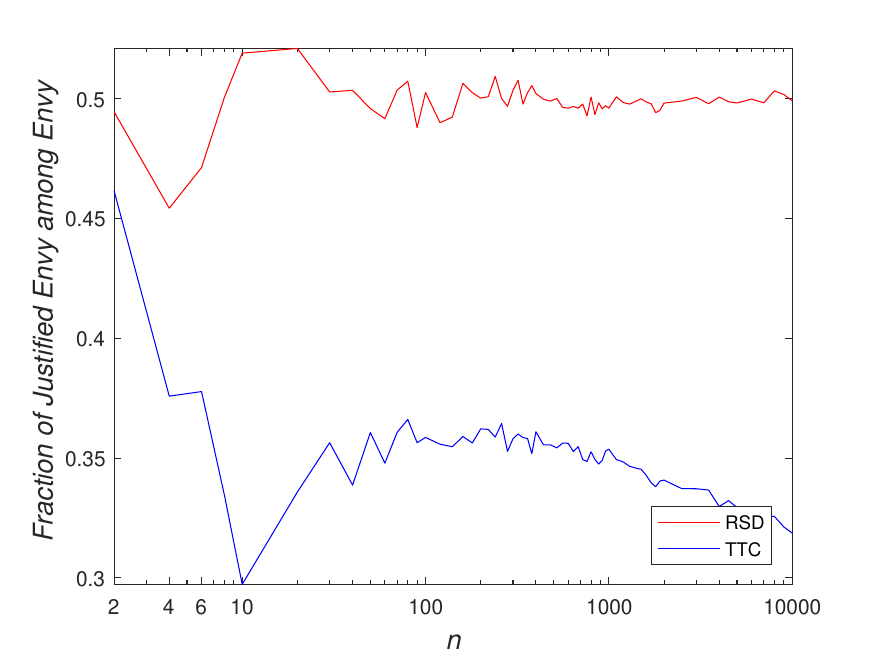} & %
\includegraphics[width=8cm,height=8cm]{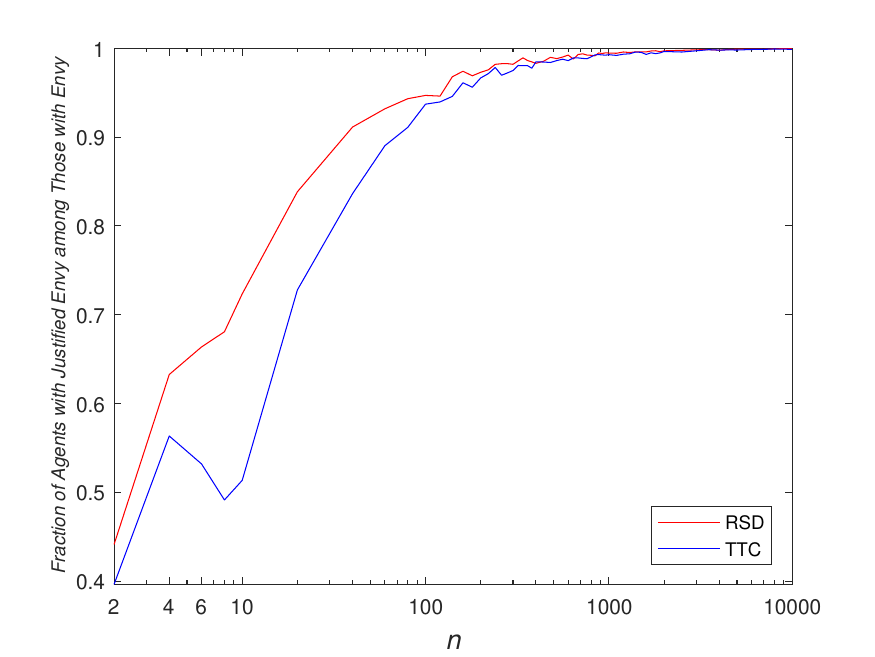}%
\tabularnewline {\footnotesize{}{}{}(a) Incidences of justified envy among
incidences of envy. } & {\footnotesize {}{}{}(b) Agents with justified envy
among those with envy. }%
\smallskip
\end{tabular}
{\footnotesize {Note: The horizontal axis indicates the size of the market,
i.e., the number of agents. The number of objects is fixed to 10. The
figures plot averages over 100 (uniform iid) draws of agents' preferences
and objects' priorities.}
}
\caption{Difference between TTC and RSD}
\label{fig_bis}
\end{figure}

In this scenario, the proportion of agents assigned via short cycles under TTC no longer vanishes in probability. Consequently, as demonstrated in Panel (a) of   \Cref{fig_bis}, the performance gap between TTC and RSD in eliminating total incidences of justified envy remains persistent, even as the market grows large.

Interestingly, convergence is restored if we evaluate performance using an alternative
metric: the ratio of agents who experience justified envy out of all agents who experience any form of envy.\footnote{Specifically, we count the number of agents experiencing justified envy and divide it by the total number of agents experiencing envy (justified or otherwise). In short, the first measure (Panel a) captures the \emph{intensity} of justified envy by counting all individual incidences, whereas this alternative measure (Panel b) strictly counts the \emph{number of agents} holding justified envy, regardless of how many peers they justifiably envy.} Panel (b) of   \Cref{fig_bis} confirms that under this alternative measure, the difference between the two mechanisms vanishes asymptotically.\footnote{We also observe asymptotic convergence when measuring the fraction of blocking pairs.}


\section{Conclusion}
\label{sec:Conclusion}

 In this paper, we have shown that the priority-enforcing power of the Top Trading Cycles mechanism vanishes in a canonical large market. As the number of agents and objects grows large, the assignment produced by TTC generates virtually the same incidence of justified envy  (e.g., the total count of blocking pairs across the entire market)  as Random Serial Dictatorship (RSD)—a mechanism that is entirely blind to priorities.   By establishing that TTC’s explicit use of priorities becomes asymptotically irrelevant, our results challenge the conventional rationale for recommending TTC as the leading Pareto-efficient and strategy-proof mechanism in environments where respecting institutional priorities or mitigating justified envy is an important policy objective.

 However, this asymptotic irrelevance should not necessarily be interpreted as a design flaw of TTC. Recall from \cite{acprt:17} that TTC is justified-envy minimal among Pareto-efficient and strategy-proof mechanisms. In light of this result, the vanishing role of priorities may instead reflect the fundamental tension created by simultaneously requiring Pareto efficiency and strategy-proofness. In large markets, these two requirements may simply leave too little scope for priorities to play a meaningful role.\footnote{Indeed, in our environment, one can find a Pareto-efficient but non-strategy-proof mechanism—such as the Boston mechanism—that generates significantly less justified envy than RSD, even in a large market.} 

Our findings suggest two natural avenues for future research. First, it remains an important theoretical and empirical task to systematically delineate the priority structures and market characteristics---such as the ``large school'' environments featuring high capacity-to-program ratios---under which TTC retains a significant advantage over RSD in eliminating justified envy. Second, our results underscore the severe tension inherent in demanding strict Pareto efficiency, exact strategy-proofness, and the elimination of justified envy. Existing mechanisms, such as TTC and Gale-Shapley's Deferred Acceptance, strictly satisfy two of these desiderata but can perform poorly on the third. To the extent that a significant loss of fairness is the inevitable cost of strictly enforcing efficiency and strategy-proofness, relaxing these requirements may yield more acceptable outcomes. Indeed, \cite{che/tercieux:15} demonstrate that it is possible to design mechanisms that achieve efficiency, no justified envy, and truth-telling \emph{approximately} as the economy grows large. For policymakers and market participants alike, embracing such asymptotic approximations may offer a highly promising alternative to the rigid trade-offs imposed by traditional mechanisms.



\section*{Appendix}
\appendix

\numberwithin{theorem}{section}
\numberwithin{lemma}{section}
\numberwithin{proposition}{section}
\numberwithin{corollary}{section}
\numberwithin{remark}{section}

\section{Proof of  \Cref{Markov}}

\label{sec:Markov}

Given the nature of conditioning mentioned earlier, it is crucial for our
purpose to keep track of the agents and objects that \emph{can} draw their
partners at random and those who \emph{cannot} in each round of TTC. This
requires us to investigate the probabilistic structure known as random
rooted forests.

To begin, consider any two finite sets $I$ and $O$, with cardinalities $%
|I|=n, |O|=o$. A \textbf{bipartite digraph} $G=(I\cup O, E)$ consists of
vertices $I$ and $O$ on two separate sides and directed edges $E\subset
(I\times O)\cup (O\times I)$, comprising ordered pairs of the form $(i,o)$
or $(o,i)$ (corresponding to edge originating from $i$ and pointing to $o$
and an edge from $o$ to $i$, respectively). A \textbf{rooted tree} is a {%
bipartite digraph that consists of a ``root'' vertex $u$ and a collection of
vertices such that there is exactly one directed path from each of the
latter vertices to the root.}\footnote{%
Sometimes, a tree is defined as an acyclic undirected connected graph. In
such a case, a tree is rooted when we name one of its vertices a ``root.''
Starting from such a rooted tree, if all edges now have a direction leading
toward the root, then the out-degree of any vertex (except the root) is 1.
So the two definitions are actually equivalent.} A \textbf{rooted forest} is
a bipartite graph which consists of a collection of disjoint rooted trees. A
\textbf{spanning rooted forest over }$I\cup O$ is a forest comprising
vertices $I\cup O$.
From now on, a spanning forest will be understood as being over $I\cup O$.

\subsection{Markov Properties of Spanning Rooted Forests Induced by TTC}

We first record the deferred-ranking observation that connects the iid primitive rankings to the forest process used below. At the beginning of round $t$, let $I_t$ and $O_t$ be the agents and objects still present, and let $\mathcal H_t$ contain the sets removed before $t$ and all comparisons already used to identify the outgoing edges that remain in the current forest. Conditional on $\mathcal H_t$ and on the current forest, if a remaining agent $i$ is a root, then her most preferred object in $O_t$ is uniform over $O_t$; these choices are independent across root agents. Symmetrically, if a remaining object $o$ is a root, then its highest-priority agent in $I_t$ is uniform over $I_t$, independently across root objects and independently of the root agents' choices. This is the standard deferred-decision property of independent uniform permutations: the history has revealed only comparisons needed for earlier outgoing edges or alternatives no longer available, so the unqueried relative order over currently available alternatives remains uniform.

The preceding observation justifies representing TTC by a random sequence of spanning rooted
forests. {More precisely, consider the beginning of each round, when the
directed edges to the agents and objects that are assigned in the previous
round have been removed, but the agents/objects have not yet pointed to
their preferred remaining agents/objects. The remaining vertices and
directed edges form a spanning rooted forest: the roots are the remaining
agents/objects who had pointed in the previous round to objects/agents that
have been assigned by the end of the previous round. The edges emanating from
these vertices have been removed, so they become roots. The other
agents/objects point to the objects/agents that still remain, so the edges
emanating from them remain intact. Clearly, the resulting graph constitutes
a spanning rooted forest. Specifically, at the beginning of round 1 of TTC,
we have the trivial forest consisting of $|I|+|O|$ roots, or isolated
vertices without edges. In the example of  \Cref{sec:Markov Property},
there are 6 separate trees: $\{1\},\{2\},\{3\},\{a\},\{b\},\{c\}$. During
round 1, each vertex in $I$ randomly points to a vertex in $O$ and each
vertex in $O$ randomly points to a vertex in $I$. Suppose agents 1 and 3
point to $c$, and agent 2 points to $a$, and all objects point to 3. Once we
delete the realized cycles ($3-c$ in the example) and the edges pointing to
them, we again get a spanning rooted forest, which consists of 3 rooted
trees: $\{1\}, \{2\to a \}, \{b\}$.} Again, the objects that are roots
randomly point to a remaining individual and the individuals that are roots
randomly point to a remaining object. Once cycles are cleared we again
obtain a rooted forest and the process repeats in the same manner.

Formally, the random sequence of forests, $F_{1},F_{2},....$, is defined as
follows. First, we let $F_{1}$ be a trivial unique forest consisting of $%
|I|+|O|$ trees with isolated vertices, forming their own roots. For any $%
i=2,...$, we first create a random directed edge from each root of $F_{i-1}$
to a vertex on the other side, and then delete the resulting cycles (these
are the agents and objects assigned in round $i-1$) and $F_{i}$ is defined
to be the resulting rooted forest. Note that this random sequence of
spanning rooted forests is a Markov chain.

For any rooted forest $F_i$, let $N_i=I_i\cup O_i$ be its vertex set and $%
k_i=(k_i^I, k_i^O)$ be the vector denoting the numbers of roots on both
sides, and use $(N_i, k_i)$ to summarize this information. And let $\mathcal{%
F}_{N_{i},k_{i}}$ denote the set of all rooted forests having $N_{i}$ as the
vertex set and $k_{i}$ as the vector of its root numbers.

Given $N_{i},k_{i}$ and $N_{i+1},k_{i+1}$, for each forest $F\in \mathcal{F}%
_{N_{i+1},k_{i+1}}$, one can compute the number of possible pairs $%
(F^{\prime },\phi )$ that could have given rise to $F$, where $F^{\prime
}\in \mathcal{F}_{N_{i},k_{i}}$ and $\phi $ maps the roots of $F^{\prime }$
in $I_{i}$ to its vertices in $O_{i}$ as well as the roots of $F^{\prime }$
in $O_{i}$ to its vertices in $I_{i}$. In words, such a pair $(F^{\prime
},\phi)$ corresponds to a set $N_i$ of agents and objects remaining at the
beginning of round $i$ of TTC, of which $k_i^I$ agents of $I_i$ and $k_i^O$
objects have lost their favorite parties (and thus they must repoint to new
partners in $N_i$ under TTC in round $i$), and the way in which they repoint to vertices available at the start of round $i$ on the opposite side, some of which may be eliminated in that round, causes the new forest $F$ to
emerge at the beginning of round $i+1$ of TTC (after cycles are cleared).
There are typically multiple such pairs $(F^{\prime },\phi )$ that could
give rise to $F$. Let $\beta (I_i,O_i,k_i^I, k_i^O
;I_{i+1},O_{i+1},k_{i+1}^I, k_{i+1}^O )$ denote the number of pairs $%
(F^{\prime },\phi )$, $F^{\prime }\in \mathcal{F}_{N_{i},k_{i}}$, causing $F$
to arise.

The first important observation we make is that $\beta$ does not depend on
the particular $F\in \mathcal{F}_{N_{i+1},k_{i+1}}$. That is, its dependence
on sets $(I_i,O_i, ;I_{i+1},O_{i+1})$ is only through their
cardinalities---specifically, $|I_{i}|,\ |O_{i}|$, ${k}_{i}$, $k_{i+1}$, $%
|I_{i+1}|$, and $|O_{i+1}|$, or equivalently $|I_{i}|,\ |O_{i}|$, ${k}_{i}$,
$k_{i+1}$, and $|N_{i}|-|N_{i+1}|$.

In order to show that this number does not depend on the particular $F\in
\mathcal{F}_{N_{i+1},k_{i+1}}$, let us, for any given $F\in \mathcal{F}%
_{N_{i+1},k_{i+1}}$, construct all such pairs. To build any such pair $%
(F^{\prime },\phi )$, one can simply choose a quadruplet $(a,b,c,d)$ of four
non-negative integers with $a+c=k_i^I$ and $b+d=k_i^O$ and proceed by,


\begin{enumerate}
\item[(i)] choosing $c$ old roots from $I_{i+1}$, and similarly, $d$ old
roots from $O_{i+1}$,

\item[(ii)] choosing $a$ old roots from $I_{i}\backslash I_{i+1}$ and
similarly, $b $ old roots from $O_{i}\backslash O_{i+1}$,

\item[(iii)] choosing a partition into cycles of $N_{i}\backslash N_{i+1}$,
each cycle of which contains at least one old root from (ii),\footnote{%
Within round $i$ of TTC, one cannot have a cycle creating only with vertices
that are not roots in the forest obtained at the beginning of round $i$.
This is due to the simple fact that a forest is an acyclic graph. Thus, each
cycle creating must contain at least one old root. Given that, by
definition, these roots are eliminated from the set of available nodes in
round $i+1$, these old roots that each cycle must contain must be from (ii).}

\item[(iv)] choosing a mapping of the $k_{i+1}^I+k_{i+1}^O$ new roots to $%
N_{i}\backslash N_{i+1}$.\footnote{%
Since, by definition, any root in $F\in \mathcal{F}_{N_{i+1},k_{i+1}}$ does
not point, this means that, in the previous round, this node was pointing to
another node which was eliminated at the end of that round.}
\end{enumerate}

Clearly, the number of pairs $(F^{\prime },\phi )$, $F^{\prime }\in \mathcal{%
F}_{N_{i},k_{i}}$, satisfying the above restrictions depends only on $%
|I_{i}|,\ |O_{i}|$, ${k}_{i}$, $k_{i+1}$, and $|N_{i}|-|N_{i+1}|$.\footnote{%
Recall that by definition of TTC, whenever a cycle creates, the same number
of individuals and objects must be eliminated in this cycle. Hence, $%
|O_{i}|-|O_{i+1}|=|I_{i}|-|I_{i+1}|$ and $%
|N_{i}|-|N_{i+1}|=2(|I_{i}|-|I_{i+1}|)$.} Accordingly, we now write this
number as $\beta (|I_{i}|,|O_{i}|, {k}_{i};|N_{i}|-|N_{i+1}|,k_{i+1})$. (We
derive $\beta$ explicitly in \Cref{sec:lem:beta} of Online Appendix.)

With this observation in hand, we can prove the following result.

\begin{lemma}
\label{lem:uniform} \label{uniform} Given $(N_j, k_j), j=1, ..., i$, every
(rooted) forest of $\mathcal{F}_{N_i, k_i}$ is equally likely.
\end{lemma}

\begin{proof}
We prove this result by induction on $i$. Since for $i=1$, by construction,
	the trivial forest is the unique forest which can occur, this is trivially
	true for $i=1$. Fix $i\geq 1$, and assume our statement is true for $i$. Let us show that it holds for $i+1$.

	Fix $N_{i}=I_{i}\cup O_{i}\supset N_{i+1}=I_{i+1}\cup O_{i+1},$
	and $k_i$ and $k_{i+1}$.   For each forest  $F\in
	\mathcal{F}_{N_{i+1},k_{i+1}}$, we consider a possible pair $%
	(F^{\prime },\phi )$ that could have given rise to $F$,  where $F^{\prime }\in \mathcal{F}_{N_{i},k_{i}}$ and $%
	\phi $ maps the roots of $F^{\prime }$ in $I_{i}$ to its vertices in $O_{i}$
	as well as the roots of $F^{\prime }$ in $O_{i}$ to its vertices in $I_{i}$.
	As we already argued, each forest $F\in
	\mathcal{F}_{N_{i+1},k_{i+1}}$ arises from the same number of such pairs---i.e., that the number of pairs $
	(F^{\prime },\phi )$, $F^{\prime }\in \mathcal{F}_{N_{i},k_{i}}$, causing $F$ to arise does not depend on the particular $F\in \mathcal{F}_{N_{i+1},k_{i+1}}$.

	The number of such pairs is given by $\beta
	(|I_{i}|,|O_{i}|,{k}_{i};|N_{i}|-|N_{i+1}|,k_{i+1})$.
	Let $\phi _{i}=(\phi _{i}^{I},\phi _{i}^{O})$ where $\phi _{i}^{I}$ is
	the random mapping from the roots of $F_{i}$ in $I_{i}$ to $O_{i}$ and $\phi
	_{i}^{O}$ is the random mapping from the roots of $F_{i}$ in $O_{i}$ to $%
	I_{i}$. Let $\phi =(\phi ^{I},\phi ^{O})$ be a generic mapping of that sort.
	Since, conditional on $F_{i}=F^{\prime }$, the mappings $\phi _{i}^{I}\ $and
	$\phi _{i}^{O}$ are uniform, we get
	\begin{align}\label{conditional}
	\Pr(F_{i+1}=F|F_{i}=F^{\prime })=&\frac{1}{|O_{i}|^{k_{i}^{I}}}\frac{1}{%
		|I_{i}|^{k_{i}^{O}}}\sum_{\phi }\Pr(F_{i+1}=F|F_{i}=F^{\prime },\phi _{i}=\phi
	) . \end{align}%
	
	Therefore, we obtain
	\begin{align}
	&\Pr (F_{i+1} =F|(N_{1},k_{1}),...,(N_{i},k_{i}))\cr
	=& \sum_{F'\in \mathcal{F}_{N_i, k_i}} \Pr (F_{i+1} =F, F_i=F'|(N_{1},k_{1}),...,(N_{i},k_{i}))\cr
	=& \sum_{F'\in \mathcal{F}_{N_i, k_i}} \Pr (F_{i+1} =F|(N_{1},k_{1}),...,(N_{i},k_{i}),F_i= F') \Pr(F_i=F'|(N_1, k_1), ...., (N_{i}, k_{i}))\cr
	=&\frac{1}{|\mathcal{F}%
		_{N_{i},k_{i}}|}\sum_{F^{\prime }\in \mathcal{F}%
		_{N_{i},k_{i}}}\Pr(F_{i+1}=F|F_{i}=F^{\prime }) \cr
	=& \frac{1}{|\mathcal{F}%
		_{N_{i},k_{i}}|}\sum_{F^{\prime }\in \mathcal{F}%
		_{N_{i},k_{i}}} \frac{1}{|O_{i}|^{k_{i}^{I}}}\frac{1}{%
		|I_{i}|^{k_{i}^{O}}}\sum_{\phi }\Pr(F_{i+1}=F|F_{i}=F^{\prime },\phi _{i}=\phi
	) \cr
	=& \frac{1}{|\mathcal{F}%
		_{N_{i},k_{i}}|}\frac{1}{|O_{i}|^{k_{i}^{I}}}\frac{1}{%
		|I_{i}|^{k_{i}^{O}}}\sum_{F^{\prime }\in \mathcal{F}%
		_{N_{i},k_{i}}} \sum_{\phi }\Pr(F_{i+1}=F|F_{i}=F^{\prime },\phi _{i}=\phi
	) \cr
	=&\frac{1}{|\mathcal{F}_{N_{i},k_{i}}|}\frac{1}{|O_{i}|^{k_{i}^{I}}}\frac{1}{%
		|I_{i}|^{k_{i}^{O}}}\beta (|I_{i}|,|O_{i}|, {k}_{i};|N_{i}|-|N_{i+1}|,k_{i+1}), \label{conditional1}  \end{align}
	where the third equality follows from  the Markov property of
	$\{F_{j}\}$ and the induction hypothesis,   the fourth follows from (\ref{conditional}), and the last follows from the definition of $\b$ and
	from the fact that the conditional probability in the sum of the penultimate line is $1$ or $0$, depending upon
	whether the forest $F$ arises from the pair $(F^{^{\prime }},\phi )$ or not. Note that this  probability is independent of $F\in \mathcal{F}_{N_{i+1},k_{i+1}}$.  Hence,
	\begin{align*}
	&\Pr(F_{i+1}=F|(N_{1},k_{1}),...,(N_{i},k_{i}),(N_{i+1},k_{i+1}))\cr
	=& \frac{\Pr(F_{i+1} =F|(N_{1},k_{1}),...,(N_{i},k_{i}))}{\Pr(F_{i+1} \in \mathcal{F}%
		_{N_{i+1},k_{i+1}}|(N_{1},k_{1}),...,(N_{i},k_{i}))}\cr
	=& \frac{\Pr(F_{i+1} =F|(N_{1},k_{1}),...,(N_{i},k_{i}))}
	{\sum_{\tilde F\in \mathcal{F}%
			_{N_{i+1},k_{i+1}}}    \Pr(F_{i+1}=\tilde F|(N_{1},k_{1}),...,(N_{i},k_{i}))}
	=   \frac{1}{|\mathcal{F}_{N_{i+1},k_{i+1}}|}, \end{align*}
	which proves that, given $(N_j, k_j), j=1, ..., i+1$, every rooted forest of $%
	\mathcal{F}_{N_{i+1}, k_{i+1}}$ is equally likely.
\end{proof}

The next lemma then follows easily.

\begin{lemma}
\label{lem:Nk} \label{lem:markov0} Random sequence $(N_i, k_i)$ forms a
Markov chain.
\end{lemma}

\begin{proof}
By (\ref{conditional1}) we must have
	\begin{align*}
	\Pr( (N_{i+1}, k_{i+1}) | (N_1, k_1), ...., (N_i, k_i))
	= &\sum_{F\in \mathcal{F}_{N_{i+1},k_{i+1}}} \Pr(F_{i+1}=F|(N_1, k_1), ...., (N_i, k_i))\\
	= &\sum_{F\in \mathcal{F}_{N_{i+1},k_{i+1}}}  \frac{1}{|\mathcal{F}_{N_{i},k_{i}}|}\frac{1}{|O_{i}|^{k_{i}^{I}}}\frac{1}{%
		|I_{i}|^{k_{i}^{O}}}\beta (|I_{i}|,|O_{i}|, {k}%
		_{i};|N_i|-|N_{i+1}|,k_{i+1}).
	\end{align*}
	Observing that the conditional probability depends only on $(N_{i+1}, k_{i+1})$ and $(N_i, k_i)$, the Markov chain property is established.
\end{proof}

The proof of  \Cref{lem:markov0} reveals in fact that the conditional
probability of $(N_{i+1}, k_{i+1})$ depends on $N_i$ only through its
cardinalities $(|I_i|, |O_i|)$, leading to the following conclusion. Let $%
n_i:=|I_i|$ and $o_i:=|O_i|$.

\begin{corollary}
\label{cor:markov} Random sequence $\{(n_i, o_i, k_i^I, k_i^O)\}$ forms a
Markov chain.
\end{corollary}

\begin{proof}
By symmetry, given $(n_1, o_1, k_1^I, k_1^O), ..., (n_i, o_i, k_i^I, k_i^O)$, the set $(I_i, O_i)$ is chosen uniformly at random among all the ${n \choose n_i}{o\choose o_i}$ possible sets.  Hence,
	\begin{align*}& \Pr( (n_{i+1}, o_{i+1}, k_{i+1}^I, k_{i+1}^O) | (n_1, o_1, k_1^I, k_1^O), ...., (n_i, o_i, k_i^I, k_i^O)) \\
	=& \sum_{(I_i, O_i): |I_i|=n_i, |O_i|=o_i}  \Pr\{(n_{i+1}, o_{i+1}, k_{i+1}^I, k_{i+1}^O)  | (n_1, o_1, k_1^I, k_1^O), ...., (n_i, o_i, k_i^I, k_i^O), (I_i, O_i)\}\\
	& \qquad  \qquad  \times \Pr\left\{(I_i,O_i)\left |\right. (n_1, o_1, k_1^I, k_1^O), ..., (n_i, o_i, k_i^I, k_i^O)\right\}\\
	=&\left( \sum_{(I_i, O_i): |I_i|=n_i, |O_i|=o_i}  \Pr\{(n_{i+1}, o_{i+1}, k_{i+1}^I, k_{i+1}^O)  | (n_1, o_1, k_1^I, k_1^O), ...., (I_i, O_i, k_i^I, k_i^O)\}\right)\frac{1}{{n \choose n_i}{o\choose o_i}}\\
	= &\left(  \sum_{{(I_i, O_i): |I_i|=n_i, |O_i|=o_i\atop (I_{i+1}, O_{i+1}): |I_{i+1}|=n_{i+1}, |O_{i+1}|=o_{i+1}}}  \Pr\{(I_{i+1}, O_{i+1}, k_{i+1}^I, k_{i+1}^O) | (n_1, o_1, k_1^I, k_1^O), ...., (I_i, O_i, k_i^I, k_i^O)\}\right) \\
	&\qquad \times \frac{1}{{n \choose n_i}{o\choose o_i}}\\
	= &\frac{1}{{n \choose n_i}{o\choose o_i}} \sum_{{(I_i, O_i): |I_i|=n_i, |O_i|=o_i\atop (I_{i+1}, O_{i+1}): |I_{i+1}|=n_{i+1}, |O_{i+1}|=o_{i+1}}}  \Pr\{(I_{i+1}, O_{i+1}, k_{i+1}^I, k_{i+1}^O) |  (I_i, O_i, k_i^I, k_i^O)\},
	\end{align*}
	where the second equality follows from the above reasoning and the last equality follows from the Markov property of $\{ (I_i, O_i, k_i^I, k_i^O)\}$. The proof is complete by the fact that the last line, as shown in the proof of  \Cref{lem:markov0}, depends only on $ (n_{i+1}, o_{i+1}, k_{i+1}^I, k_{i+1}^O) ,  (n_i, o_i, k_i^I, k_i^O))$.
\end{proof}

\subsection{Proof of  \Cref{Markov}}



We first compute the probability of transition from $%
(n_{i},o_{i},k_{i}^{I},k_{i}^{O})$ to $%
(n_{i+1},o_{i+1},k_{i+1}^{I},k_{i+1}^{O})$ such that $k_{i+1}^{I}=\lambda^I$
and $k_{i+1}^{O}=\lambda^O$:
\begin{align*}
&\mathbf{P}(n,o, k_{i}^{I}, k_{i}^{O} ;m,\lambda^I,\lambda^O) \\
:=& \Pr
\left\{n_{i}-n_{i+1}=o_{i}-o_{i+1}=m,k_{i+1}^{I}=\lambda^{I},k_{i+1}^{O}=%
\lambda ^{O}\, \left|\right. \, n_{i}=n,o_{i}=o,k_{i}^{I}, k_{i}^{O}
\right\}.
\end{align*}

This will be computed as a fraction $\frac{\Theta }{\Upsilon }.$ The
denominator $\Upsilon $ counts the number of rooted forests in the bipartite
digraph with $k_{i}^{I}$ roots in $I_{i}$ and $k_{i}^{O}$ roots in $O_{i}$,
multiplied by the ways in which $k_{i}^{I}$ roots of $I_{i}$ could point to $%
O_{i}$ and $k_{i}^{O}$ roots of $O_{i}$ could point to $I_{i}$.\footnote{%
Given that we have $n_{i}=n\ $individuals, $o_{i}=o\ $objects and $k_{i}^{I}$
roots in $I_{i}$ and $k_{i}^{O}$ roots in $O_{i}$ at the beginning of step $%
i $ under TTC, one may think of this as the total number of possible
bipartite digraph one may obtain via TTC at the end of step $i$ when we let $%
k_{i}^{I}$ roots in $I_i$ point to their remaining most favorite object and $%
k_{i}^{O}\ $ roots in $O_i$ point to their remaining most favorite
individual.} We let $f(n,o,k_i^I,k_i^O)$ denote the number of rooted forests
in a bipartite digraph (with $n$ and $o$ vertices on both sides) containing $%
k_i^I$ and $k_i^O$ roots on both sides. Thus, since there are $%
o^{k_i^I}n^{k_i^O}$ ways in which $k_i^I$ roots in $I_{i}$ point to $O_{i}$
and $k_i^O$ roots in $O_{i}$ could point to $I_{i}$, $\Upsilon =
o^{k_i^I}n^{k_i^O}f(n,o,k_i^I,k_i^O)$.

The numerator $\Theta $ counts the number of ways in which $m$ agents are
chosen from $I_{i}$ and $m$ objects are chosen from $O_{i}$ to form a
bipartite bijection with the possible locations of old roots enumerated through $\beta$ and with each resulting cycle containing at least one old root, and for each such choice, the number of ways in
which the remaining vertices form a spanning rooted forest and the $\lambda
^I$ roots in $I_{i+1}$ point to objects in $O_{i}\setminus O_{i+1}$ and $%
\lambda ^O$ roots in $O_{i+1}$ point to agents in $I_{i}\setminus I_{i+1}$.
To compute this, recall $\beta (n, o, k_i^I, k_i^O; m, \lambda^I, \lambda^O),
$ denotes, for any $F$ with $n-m$ agents and $o-m$ objects and roots $%
\lambda^I $ and $\lambda^O$ on both sides, the total number of pairs $%
(F^{\prime }, \phi)$ that could have given rise to $F$, where $F^{\prime }$
has $n$ agents and $o$ objects with $(k_i^I,k_i^O)$ roots and $\phi$ maps
the roots to current opposite-side vertices at the start of round $i$.

The numerator $\Theta$ is now computed as $\Theta ={\binom{n}{m}}{\binom{o}{m%
}}f(n-m,o-m,\lambda ^I,\lambda ^O)\beta (n,o,k_i^I, k_i^O ;m,\lambda
^I,\lambda ^O)$.

In  \Cref{sec:lem:beta} of Online Appendix, we explicitly
compute $\beta (n,o,k_i^I, k_i^O ;m,\lambda ^I,\lambda ^O)$, then,
collecting terms, we show in \Cref{sec:computation
transition} that
\begin{equation*}
\mathbf{P}(n,o,k_i^I,k_i^O ;m,\lambda ^I,\lambda ^O)=\frac{\Theta}{\Upsilon}
=\frac{1}{o^{n}n^{o}}\left( \frac{ n!}{(n-m)!}\right) \left( \frac{o!}{(o-m)!%
}\right) m^{\lambda ^I+\lambda ^O}f(n-m,o-m,\lambda ^I,\lambda ^O).
\end{equation*}
A key observation is that this expression does not depend on $(k_i^I,k_i^O)$%
, which, together with  \Cref{cor:markov}, implies that $%
(n_{i},o_{i})$ forms a Markov chain.

Its transition probability can be derived by summing the expression over all
possible $(\lambda ^I,\lambda ^O)$'s:
\begin{equation*}
p_{n,o;m}:=\sum_{0\leq \lambda ^I\leq n-m,0\leq \lambda ^O\leq o-m} \mathbf{P%
}(n,o,k_i^I,k_i^O ;m,\lambda ^I,\lambda ^O)\text{.}
\end{equation*}
which is equal to the formula stated in the lemma $\left(\frac{m}{(on)^{m+1}}%
\right)\left(\frac{n!}{(n-m)!} \right)\left( \frac{o!}{(o-m)!}
\right)(o+n-m) $ as shown in  \Cref{sec:computation
transition}.

\section{Proof of  \Cref{prop:stopping}}

\label{sec:stopping}

We establish several lemmas before proving the proposition. Assume wlog that
$|I|=n\leq o=|O|$ (the proof is symmetric when $n\geq o$). Let $\{n_{t}\}$
be the (random) sequence corresponding to the number of individuals at Step $%
t$ of TTC. By our main result, this is a Markov chain. Let $c_{t}$ be the
number of cyclic vertices on the individual side obtained in the graph of
TTC at Step $t$ so that $n_{t+1}=n_{t}-c_{t}$ for each $t\geq 1$. In
general, $n_{t}=n-\sum_{k=1}^{t-1}c_{k}$. Thus, $\mathbb{E}%
[n_{t}]=n-\sum_{k=1}^{t-1}\mathbb{E}[c_{k}]$. Finally, letting $\{o_{t}\}$
be the (random) sequence corresponding to the number of objects at Step $t$
of TTC, we observe that $n_{t}\leq o_{t}$ for all $t\geq 1$ since $n\leq o$
(and the same number of individuals and objects are assigned at each step).

The following lemma shows that, if we start from any Step $t_{0}$ of TTC
with a number of agents/objects $n_{t_{0}}\geq \delta n$, for any
arbitrarily small $\delta>0$, then with a significant probability, after a
number of steps linear in $\sqrt{n}$ we will end up with fewer than $\delta
n $ agents remaining in the market.

\begin{lemma}\label{lem:stopping}
Consider any Step $t_{0}\geq 1$ of TTC. Fix any $\delta >0$ and let $c:=\frac{2}{\sqrt{\pi \delta }}$. There is $\gamma >0$ such that $\liminf \Pr
\{n_{t_{0}+c\sqrt{n}}\leq \delta n\left \vert n_{t_{0}}\geq \delta n\right.
\}>\gamma $ where $\gamma $ does not depend on $t_{0}$.
\end{lemma}

\begin{proof}   In the sequel, we condition on the event that $n_{t_{0}}\geq \delta n$. By the Markov chain property, we can view the process as starting with a number $n_{t_{0}}$ of agents/objects. To avoid notational clutter, we suppress the dependence on this conditioning event. Assume, to get a contradiction along a subsequence, that $\Pr \{n_{t_{0}+c\sqrt{n}}>\delta n\}\to1$. Then $\Pr\{n_t>\delta n\}\to1$ uniformly for all $t_{0}\le t\le t_{0}+c\sqrt n$.

We use the cyclic-agent version of Jaworski's formula. If a round begins with $a$ agents and $b$ objects, the expected number of cyclic agents is asymptotic to $\sqrt{\pi ab/(2(a+b))}$. Hence, uniformly over states with $a\ge \delta n$ and $b\ge a$, for all sufficiently large $n$,
\[
\mathbb E[c_t\mid n_t=a,o_t=b]\ge (1-\varepsilon)\frac12\sqrt{\pi\delta n}.
\]
Because $\Pr\{n_t>\delta n\}\to1$ uniformly over the block, it follows that, for all $t_0\le t\le t_0+c\sqrt n$ and all sufficiently large $n$,
\[
\mathbb E[c_t]\ge (1-\varepsilon)\frac12\sqrt{\pi\delta}\sqrt n.
\]

Thus, for all sufficiently large $n$,
\begin{eqnarray*}
\mathbb{E}[n_{t_{0}+c\sqrt{n}}]
&=&\mathbb{E}\left[n_{t_{0}}-\sum_{k=t_{0}}^{t_{0}+c\sqrt{n}-1}c_{k}\right] \\
&\leq& n-\left(c\sqrt n\right)(1-\varepsilon)\frac12\sqrt{\pi\delta}\sqrt n \\
&=& n-(1-\varepsilon)n = \varepsilon n,
\end{eqnarray*}
Hence $\Pr\{n_{t_{0}+c\sqrt n}\le \delta n\}\to1$, contradicting the assumed subsequence. Therefore the desired probability is bounded away from zero. The bound is uniform in $t_0$ by the Markov property.
\end{proof}

\begin{lemma}\label{lem:stopping2}
Fix any $\delta >0$ and let $c:=\frac{2}{\sqrt{\pi \delta }}$. For any $\xi
>0$, for any $k\in \mathbb{N}$ large enough, $\Pr \{n_{kc\sqrt{n}}\leq
\delta n\}>\xi$ for any $n$ large enough.
\end{lemma}

\begin{proof} We know by the previous lemma that there is $\gamma >0$ such
	that for $n$ large enough, $\Pr \{n_{c\sqrt{n}}\leq \delta n\}>\gamma $.
	First, note that $\Pr \{n_{2c\sqrt{n}}\leq \delta n\}>\gamma +(1-\gamma
	)\gamma $. Indeed, because $\{n_{t}\}$ is a decreasing sequence, $\{n_{c%
		\sqrt{n}}\leq \delta n\}$ implies $\{n_{2c\sqrt{n}}\leq \delta n\}$. Hence,
	we have
	\begin{align*}
		& \Pr\{n_{2c\sqrt{n}} \leq \delta n\} \\
		&=\Pr \{n_{c\sqrt{n}}\leq \delta n\} \Pr \{n_{2c\sqrt{n}}\leq \delta n\left
		\vert n_{c\sqrt{n}}\leq \delta n \right. \}+\Pr \{n_{c\sqrt{n}}>\delta
		n\} \Pr \{n_{2c\sqrt{n}}\leq \delta n\left \vert n_{c\sqrt{n}}>\delta n
		\right. \} \\
		&=\Pr \{n_{c\sqrt{n}}\leq \delta n\}+\Pr \{n_{c\sqrt{n}}>\delta n\} \Pr
		\{n_{2c\sqrt{n}}\leq \delta n\left \vert n_{c\sqrt{n}}>\delta n \right.
		\}
	\end{align*}Applying  \Cref{lem:stopping} for $t_{0}=1$, we know that, for $n$ large enough, $\Pr
	\{n_{c\sqrt{n}}\leq \delta n\}>\gamma $. In addition, applying  \Cref{lem:stopping} for $%
	t_{0}=c\sqrt{n}$, we know that, for $n$ large enough, $\Pr \{n_{2c\sqrt{n}%
	}\leq \delta n\left \vert n_{c\sqrt{n}}>\delta n \right. \}>\gamma $.
	Thus, we obtain that for any $n$ large enough, $\Pr \{n_{2c\sqrt{n}}\leq \delta n\} \geq \gamma +(1-\gamma
	)\gamma $, as claimed.
	
	Similar reasoning yields that for each $k\in \mathbb{N}$, there is $N$ large
	enough so that%
	\begin{equation*}
	\Pr \{n_{kc\sqrt{n}}\leq \delta n\}>\sum_{\ell =1}^{k}(1-\gamma )^{\ell
		-1}\gamma =1-(1-\gamma )^{k}\text{.}
	\end{equation*}%
	Note that the right-hand side is equal to the cumulative distribution at $k$
	of a geometric distribution with parameter $\gamma $. Clearly, this goes to $%
	1$ as $k$ increases and so our argument is complete. Thus, if we fix any $%
	\xi >0$, we can find $k$ large enough so that $\Pr \{n_{kc\sqrt{n}}\leq
	\delta n\}>\xi $ for any $n$ large enough, as was to be proved.
\end{proof}

We are now ready to prove  \Cref{prop:stopping}.

\begin{proof}
	Fix any $\alpha >0$ and $\xi <1$, we claim that there is $n$
	large enough so that $\Pr \{ \frac{T}{n}\leq \alpha \}>\xi $. Consider any $%
	\delta \in (0,\alpha )$ and fix $k\in \mathbb{N}$ and $c=\frac{4}{\sqrt{\pi
			\delta }}$ in order to have $\Pr \{n_{kc\sqrt{n}}\leq \delta n\}>\xi $ for any $n$ large enough;
	which is well-defined by   \Cref{lem:stopping2}. Note that $\{n_{kc\sqrt{n}}\leq \delta n\}$
	implies that $T\leq kc\sqrt{n}+\delta n$. Because, $\delta <\alpha $, the
	term on the right-hand side of the inequality is smaller than $\alpha n$
	when $n$ is large enough. Thus, for $n$ large enough, we obtain $\Pr \{n_{kc%
		\sqrt{n}}\leq \delta n\} \leq \Pr \{T\leq kc\sqrt{n}+\delta n\} \leq \Pr \{%
	\frac{T}{n}\leq \alpha \}$. Now, because, for $n$ large enough, $\Pr \{n_{kc\sqrt{n}}\leq
	\delta n\}>\xi $, we obtain that for $n$ large enough, $\Pr \{ \frac{T}{n}%
	\leq \alpha \}>\xi $, as claimed.
\end{proof}

\section{The Number of Objects Assigned via Short Cycles}

Recall the random sequence of forests, $F_{1},F_{2},....$ is defined in the
main text, where $F_{1}$ is a trivial unique forest consisting of $|I|+|O|$
trees with isolated vertices, forming their own roots. For any $i=2,...$, we
first create a random directed edge from each root of $F_{i-1}$ to a vertex
on the other side, and then delete the resulting cycles (these are the
agents and objects assigned in round $i-1$) and $F_{i}$ is defined to be the
resulting rooted forest.

We begin with the following question: If round $k$ of TTC begins with a
rooted forest $F$, what is the expected number of short-cycles that will
form at the end of that round? We will show that, irrespective of $F$, this
expectation is bounded by $2$. To show this, we will make a couple of
observations.

To begin, let $n_{k}$ be the cardinality of the set $I_{k}$ of individuals
in our forest $F$ and let $o_{k}$ be the cardinality of $O_{k}$, the set of $%
F$'s objects. And, let $A\subset I_k$ be the set of roots on the individuals
side of our given forest $F$ and let $B\subset O_k$ be the set of its roots
on the objects side. Let their cardinalities be $a$ and $b$, respectively.

Now, observe that for any $(i,o)\in A\times B$, the probability that $(i,o)$
forms a short-cycle is $\frac{1}{n_{k}}\frac{1}{o_{k}}$. For any $(i,o)\in
(I_{k}\backslash A)\times B$, the probability that $(i,o)$ forms a
short-cycle is $\frac{1}{n_{k}}$ if $i$ points to $o$ and $0$ otherwise.
Similarly, for $(i,o)\in A\times (O_{k}\backslash B)$, the probability that $%
(i,o)$ forms a short-cycle is $\frac{1}{o_{k}}$ if $o$ points to $i$ and $0$
otherwise. Finally, for any $(i,o)\in (I_{k}\backslash A)\times
(O_{k}\backslash B)$, the probability that $(i,o)$ forms a short-cycle is $0$
(by definition of a forest, $i$ and $o$ cannot be pointing to each other in
the forest $F$). So, given the forest $F$, the expectation of the number $%
S_{k}$ of short-cycles is
\begin{eqnarray*}
\mathbb{E}\left[ S_{k}\vert F_k=F \right] &=&\mathbb{E}\left[ \sum_{(i,o)\in
I_{k}\times O_{k}}\mathbf{1}_{\{(i,o)\text{ is a short-cycle}\}}\bigg\vert %
F_k=F \right] \\
&=&\sum_{(i,o)\in I_{k}\times O_{k}}\mathbb{E}\left[ \mathbf{1}_{\{(i,o)%
\text{ is a short-cycle}\}}\left\vert F_k=F\right. \right] \\
&=&\sum_{(i,o)\in A\times B}\mathbb{E}\left[ \mathbf{1}_{\{(i,o)\text{ is a
short-cycle}\}}\left\vert F_k=F\right. \right] \\
&& +\sum_{(i,o)\in (I_{k}\backslash A)\times B}\mathbb{E}\left[ \mathbf{1}%
_{\{(i,o)\text{ is a short-cycle}\}}\left\vert F_k=F\right. \right] \\
&&+\sum_{(i,o)\in A\times (O_{k}\backslash B)}\mathbb{E}\left[ \mathbf{1}%
_{\{(i,o)\text{ is a short-cycle}\}}\left\vert F_k=F\right. \right] \\
&=&\sum_{(i,o)\in A\times B}\Pr \{(i,o)\text{ is a short-cycle}\left\vert
F_k=F\right. \} \\
&&+\sum_{(i,o)\in (I_{k}\backslash A)\times B}\Pr \{(i,o)\text{ is a
short-cycle}\left\vert F_k=F\right. \} \\
&&+\sum_{(i,o)\in I_{k}\times (O_{k}\backslash B)}\Pr \{(i,o)\text{ is a
short-cycle}\left\vert F_k=F\right. \} \\
&\leq &\frac{ab}{n_{k}o_{k}}+\frac{n_{k}-a}{n_{k}}+\frac{o_{k}-b}{o_{k}} \\
&=&2-\frac{ao_{k}+bn_{k}-ab}{n_{k}o_{k}}\le 2.
\end{eqnarray*}

Observe that since $o_{k}\geq b$, the above term is smaller than $2$. Thus,
as claimed, we obtain the following result.\footnote{%
Note that the bound is pretty tight: if the forest $F$ has one root on each
side and each node$\ $which is not a root points to the (unique) root on the
opposite side, the expected number of short-cycles given $F$ is $\frac{1}{%
n_{k}o_{k}}+\frac{n_{k}-1}{n_{k}}+\frac{o_{k}-1}{o_{k}}\rightarrow 2$ as $%
n_{k},o_{k}\rightarrow \infty $. Thus, the conditional expectation of $s_{k}$
is bounded by $2$ and, asymptotically, this bound is tight. However, we can
show, using a more involved computation, that the unconditional expectation
of $s_{k}$ is bounded by $1$. The details of the computation are available
upon request.}

\begin{proposition}
\label{prop: conditional exp short cycle} If TTC round $k$ begins with any forest $F$, \begin{equation*}
\mathbb{E}\left[ S_{k}\left\vert F_k=F\right. \right] \leq 2.
\end{equation*}
\end{proposition}

Given that our upper bound holds for any forest $F$, we get the following
corollary.

\begin{corollary}
\label{cor: conditional exp short cycle} For any round $k$ of TTC, $\mathbb{E}%
\left[ S_{k}\right] \leq 2.$
\end{corollary}

\section{Proof of  \Cref{thm: irrelevance}}


\subsection{Proof of  \Cref{thm: irrelevance} (a)}

\label{proof:corner stone}

Recall our observation that {the priority rank realized for each object $o$
that is assigned via a long cycle is uniform on $\{R^*_o+1,...,n\}$}. This
in turn means that the rank enjoyed by each object $o\in \bar O$ is
stochastically dominated by the uniform distribution across $\{\left \lceil
\log^{1+\varepsilon }(n)\right \rceil +1,...,n\}$. To precisely characterize
the implication of this observation, recall that $R_{o}$ (resp. $R_{i}$)
denote the rank enjoyed by object $o$ (resp. enjoyed by individual $i$)
under TTC. We also let an arbitrary vector $(x_{k})_{k\in K}$ be denoted by $%
\mathbf{x}_{K}$. For instance, $\mathbf{R}_{O}$ stands for $\{R_o\}_{o \in
O} $. We are now ready to present the cornerstone for the proof of part (a)
of   \Cref{thm: irrelevance}.

\begin{proposition}
\label{corner stone} Fix any $I^{\prime }\subseteq I$ and $O^{\prime
}\subseteq O^{\prime \prime }\subseteq O$. For any $\boldsymbol{\ell }%
_{O^{\prime }},\boldsymbol{\ell }_{I^{\prime }}$,
\begin{equation*}
\Pr \left\{ \mathbf{R}_{O^{\prime }}\leq \boldsymbol{\ell }_{O^{\prime }},%
\mathbf{R}_{I^{\prime }}\leq \boldsymbol{\ell }_{I^{\prime }}\mid \bar{O}%
=O^{\prime \prime }\right\} \geq \prod_{o\in O^{\prime }}\Pr \left\{
Y_{o}\leq \ell _{o}\right\} \Pr \left\{ \mathbf{R}_{I^{\prime }}\leq
\boldsymbol{\ell }_{I^{\prime }}\mid \bar{O}=O^{\prime \prime }\right\}
\end{equation*}%
where $\{Y_{o}\}_{o\in O^{\prime }}$ is a collection of iid random variables
where each $Y_{o}$ follows the uniform distribution over $\{\left\lceil \log
^{1+\varepsilon }(n)\right\rceil +1,...,n\}$. In addition, for any $%
\boldsymbol{\ell }_{O^{\prime }},\boldsymbol{\ell }_{I^{\prime }}$,
\begin{equation*}
\Pr \left\{ \mathbf{R}_{O^{\prime }}\leq \boldsymbol{\ell }_{O^{\prime }},%
\mathbf{R}_{I^{\prime }}\leq \boldsymbol{\ell }_{I^{\prime }}\mid \bar{O}%
=O^{\prime \prime }\right\} \leq \prod_{o\in O^{\prime }}\Pr \left\{
U_{o}\leq \ell _{o}\right\} \Pr \left\{ \mathbf{R}_{I^{\prime }}\leq
\boldsymbol{\ell }_{I^{\prime }}\mid \bar{O}=O^{\prime \prime }\right\}
\end{equation*}%
where $\{U_{o}\}_{o\in O^{\prime }}$ is a collection of iid random variables
where each $U_{o}$ follows the uniform distribution over $\{1,...,n\}$.
\end{proposition}

Roughly speaking, the proposition asserts that the distribution of the rank
enjoyed by each object within $\bar O$ is ``squeezed'' (according to
first-order stochastic dominance) in between uniform from $\{ \left \lceil
\log ^{1+\varepsilon }(n)\right \rceil +1,...,n\}$ from above and uniform
from $\{ 1,...,n\}$ from below, independently of the distribution of the
ranks enjoyed by the agents and the ranks enjoyed by the other objects in
the set $\bar O$.

To prove  \Cref{corner stone}, we start with the following lemma.

\begin{lemma}
Fix any $O^{\prime \prime }\subseteq O$. For any $O^{\prime }\subseteq
O^{\prime \prime }$, for any $\boldsymbol{\ell }_{I}:=(\ell _{i})_{i\in I},%
\boldsymbol{\ell }_{O^{\prime }}:=(\ell _{o})_{o\in O^{\prime }}$ and $%
\mathbf{R}^*_{O^{\prime }}:=(r_{o})_{o\in O^{\prime }}$,
\begin{equation*}
\Pr \{ \mathbf{R}_{I}=\boldsymbol{\ell }_{I},\mathbf{R}_{O^{\prime }}=%
\boldsymbol{\ell }_{O^{\prime }},\mathbf{R}^*_{O^{\prime }}=\mathbf{r}%
_{O^{\prime }},\bar{O}=O^{\prime \prime }\}=0
\end{equation*}%
if $\ell _{o}\leq r_{o}$ for some $o\in O^{\prime }$ and is a number, possibly zero, whose value does not depend on $\boldsymbol{\ell }_{O^{\prime }}$ otherwise.
\end{lemma}

\begin{proof}
In the sequel, to save on notation, we let $\mathcal{E}\ $be $\{ \bar{O}%
=O^{\prime \prime }\}$. We first note that
\begin{equation*}
\Pr \{ \mathbf{R}_{I}=\boldsymbol{\ell }_{I},\mathbf{R}_{O^{\prime }}=%
\boldsymbol{\ell }_{O^{\prime }},\mathbf{R}^*_{O^{\prime }}=\mathbf{r}%
_{O^{\prime }},\mathcal{E}\}=0.
\end{equation*}%
\ if for some $o\in O^{\prime }$, $\ell _{o}\leq r_{o}$. Indeed, by definition, when object $o$ is involved in a cycle, it points to the agent whose priority rank is $r_o$. In addition, $%
o\in O^{\prime }\subseteq \bar{O}$ implies that object $o$ is assigned via a
long cycle, hence, the individual $o$ is matched to must have a priority rank
strictly greater than $r_{o}$.

Now, for any $\boldsymbol{\ell }_{O^{\prime }},\boldsymbol{\ell }_{O^{\prime
}}^{\prime }\ $satisfying $\boldsymbol{\ell }_{O^{\prime }},\boldsymbol{\ell
}_{O^{\prime }}^{\prime }\gg \mathbf{r}_{O^{\prime }}$, we argue that
\begin{equation*}
\Pr \{ \mathbf{R}_{I}=\boldsymbol{\ell }_{I},\mathbf{R}_{O^{\prime }}=%
\boldsymbol{\ell }_{O^{\prime }},\mathbf{R}^*_{O^{\prime }}=\mathbf{r}%
_{O^{\prime }},\mathcal{E}\}=\Pr \{ \mathbf{R}_{I}=\boldsymbol{\ell }_{I},%
\mathbf{R}_{O^{\prime }}=\boldsymbol{\ell }_{O^{\prime }}^{\prime },\mathbf{R%
}^*_{O^{\prime }}=\mathbf{r}_{O^{\prime }},\mathcal{E}\}.
\end{equation*}%
Indeed, fix a profile of preferences and priorities yielding $\{ \mathbf{R}%
_{I}=\boldsymbol{\ell }_{I},\mathbf{R}_{O^{\prime }}=\boldsymbol{\ell }%
_{O^{\prime }},\mathbf{R}^*_{O^{\prime }}=\mathbf{r}_{O^{\prime }},\mathcal{E}%
\}$. For each object $o\in O^{\prime }$, let $i$ be the individual with rank
$\ell _{o}^{\prime }$. Swap $i$ and $k:=$ TTC$(o)$ in $o$'s priority
ordering. Clearly, $k$ has rank $\ell _{o}^{\prime }$ at $o$. In addition,
since for each object $o$, the agent to whom $o$ points under the original profile has priority rank $r_o$, which is smaller than the priority ranks of both $i$ and $k$ at the original profile (recall that by assumption $%
\boldsymbol{\ell }_{O^{\prime }},\boldsymbol{\ell }_{O^{\prime }}^{\prime
}\gg \mathbf{r}_{O^{\prime }}$), each step of the TTC algorithm remains
unchanged after the swaps. Hence, $\{ \mathbf{R}_{I}=\boldsymbol{\ell }_{I},%
\mathbf{R}_{O^{\prime }}=\boldsymbol{\ell }_{O^{\prime }}^{\prime },\mathbf{R}^*_{O^{\prime }}=\mathbf{r}_{O^{\prime }},\mathcal{E}\}$ is obtained. Thus,
we have an injection from the set of profiles of preferences and priorities
yielding $\{ \mathbf{R}_{I}=\boldsymbol{\ell }_{I},\mathbf{R}_{O^{\prime }}=%
\boldsymbol{\ell }_{O^{\prime }},\mathbf{R}^*_{O^{\prime }}=\mathbf{r}%
_{O^{\prime }},\mathcal{E}\}$ to the one yielding $\{ \mathbf{R}_{I}=%
\boldsymbol{\ell }_{I},\mathbf{R}_{O^{\prime }}=\boldsymbol{\ell }%
_{O^{\prime }}^{\prime },\mathbf{R}^*_{O^{\prime }}=\mathbf{r}_{O^{\prime }},%
\mathcal{E}\}$. Given the iid distribution of priority order, it follows that%
\begin{equation*}
\Pr \{ \mathbf{R}_{I}=\boldsymbol{\ell }_{I},\mathbf{R}_{O^{\prime }}=%
\boldsymbol{\ell }_{O^{\prime }},\mathbf{R}^*_{O^{\prime }}=\mathbf{r}%
_{O^{\prime }},\mathcal{E}\} \leq \Pr \{ \mathbf{R}_{I}=\boldsymbol{\ell }%
_{I},\mathbf{R}_{O^{\prime }}=\boldsymbol{\ell }_{O^{\prime }}^{\prime },%
\mathbf{R}^*_{O^{\prime }}=\mathbf{r}_{O^{\prime }},\mathcal{E}\}.
\end{equation*}%
\ A similar reasoning shows that the inequality holds in the other direction
as well. \end{proof}

Now, we can complete the proof of  \Cref{corner stone}. Here
again, in the sequel, to save on notation, we let $\mathcal{E}\ $be $\{\bar{O%
}=O^{\prime \prime }\}$. By the above lemma, for any $O_{1},O_{2}\subseteq
O^{\prime }$ disjoint, whenever well-defined, $\Pr \{\mathbf{R}_{O_{1}}=%
\boldsymbol{\ell }_{O_{1}},\mathbf{R}_{I^{\prime }}=\boldsymbol{\ell }%
_{I^{\prime }}\mid \mathbf{R}_{O_{2}}=\boldsymbol{\ell }_{O_{2}},\mathbf{R}%
^*_{O^{\prime }}=\mathbf{r}_{O^{\prime }},\mathcal{E}\} \ $has a value, possibly zero, which does not depend on $\boldsymbol{\ell }_{O_{2}}$.\footnote{%
Indeed, by the above lemma, $\Pr \{ \mathbf{R}_{O_{2}}=\boldsymbol{\ell }%
_{O_{2}},\mathbf{R}^*_{O^{\prime }}=\mathbf{r}_{O^{\prime }},\mathcal{E}\}$
does not depend on $\boldsymbol{\ell }_{O_{2}}$ as long as it is strictly
positive. In addition, provided that the conditional distribution is
well-defined (i.e., $\boldsymbol{\ell }_{O_{2}}\gg \mathbf{r}_{O_{2}}$), $%
\Pr \{ \mathbf{R}_{O_{1}}=\boldsymbol{\ell }_{O_{1}},\mathbf{R}_{I^{\prime
}}=\boldsymbol{\ell }_{I^{\prime }},\mathbf{R}_{O_{2}}=\boldsymbol{\ell }%
_{O_{2}},\mathbf{R}^*_{O^{\prime }}=\mathbf{r}_{O^{\prime }},\mathcal{E}\}$
is equal to $0$ if $\ell _{o}\le r_o$ for some $o\in O_{1}$. In this case, this
remains equal to $0$ irrespective of $\boldsymbol{\ell }_{O_{2}}$. Finally,
if $\boldsymbol{\ell }_{O_{1}}\gg \mathbf{r}_{O_{1}}$ then the above lemma
implies that $\Pr \{ \mathbf{R}_{O_{1}}=\boldsymbol{\ell }_{O_{1}},\mathbf{R}%
_{I^{\prime }}=\boldsymbol{\ell }_{I^{\prime }},\mathbf{R}_{O_{2}}=%
\boldsymbol{\ell }_{O_{2}},\mathbf{R}^*_{O^{\prime }}=\mathbf{r}_{O^{\prime
}},\mathcal{E}\}$ has a value, possibly zero, which does not depend on $%
\boldsymbol{\ell }_{O_{2}}$.} Hence, we can write
\begin{eqnarray}  \label{independence}
&&\Pr \left \{ \mathbf{R}_{O_{1}}=\boldsymbol{\ell }_{O_{1}},\mathbf{R}%
_{I^{\prime }}=\boldsymbol{\ \ell }_{I^{\prime }}\mid \mathbf{R}%
^*_{O^{\prime }}=\mathbf{r}_{O^{\prime }},\mathcal{E}\right \}\cr &=&\sum_{%
\boldsymbol{\ell }_{O_{2}}^{\prime }}\Pr \left \{ \mathbf{R}_{O_{2}}=%
\boldsymbol{\ell }_{O_{2}}^{\prime }\mid \mathbf{R}^*_{O^{\prime }}=\mathbf{r%
}_{O^{\prime }},\mathcal{E}\} \Pr \{ \mathbf{R}_{O_{1}}=\boldsymbol{\ell }%
_{O_{1}},\mathbf{R}_{I^{\prime }}=\boldsymbol{\ell }_{I^{\prime }}\mid
\mathbf{R}_{O_{2}}=\boldsymbol{\ell }_{O_{2}}^{\prime },\mathbf{R}%
^*_{O^{\prime }}=\mathbf{r}_{O^{\prime }},\mathcal{E}\right \} \cr &=&\Pr
\left \{ \mathbf{R}_{O_{1}}=\boldsymbol{\ell }_{O_{1}},\mathbf{R}_{I^{\prime
}}=\boldsymbol{\ell }_{I^{\prime }}\mid \mathbf{R}_{O_{2}}=\boldsymbol{\ell }%
_{O_{2}},\mathbf{R}^*_{O^{\prime }}=\mathbf{r}_{O^{\prime }},\mathcal{E}%
\right \} \sum_{\boldsymbol{\ell }_{O_{2}}^{\prime }}\Pr \left \{ \mathbf{R}%
_{O_{2}}=\boldsymbol{\ell }_{O_{2}}^{\prime }\mid \mathbf{R}^*_{O^{\prime }}=%
\mathbf{r}_{O^{\prime }},\mathcal{E}\right \} \cr &=&\Pr \left \{ \mathbf{R}%
_{O_{1}}=\boldsymbol{\ell }_{O_{1}},\mathbf{R}_{I^{\prime }}=\boldsymbol{%
\ell }_{I^{\prime }}\mid \mathbf{R}_{O_{2}}=\boldsymbol{\ell }_{O_{2}},%
\mathbf{R}^*_{O^{\prime }}=\mathbf{r}_{O^{\prime }},\mathcal{E}\right \},
\end{eqnarray}
where $\boldsymbol{\ell }_{O_{2}}$ is an arbitrary profile under which the
above conditional probability is well-defined.\ Hence, conditional on $\{%
\mathbf{R}^*_{O^{\prime }}=\mathbf{r}_{O^{\prime }}\} \ $and $\mathcal{E}$,
the joint distribution of $\mathbf{R}_{O_{1}}\ $and $\mathbf{R}_{I^{\prime
}}\ $does not depend on the specific realization of $\mathbf{R}_{O_{2}}$.
This implies first that (setting $O_{1}=\emptyset $)%
\begin{eqnarray}  \label{independence_bis}
\Pr \left \{ \mathbf{R}_{I^{\prime }}\leq \boldsymbol{\ell }_{I^{\prime
}}\mid \mathbf{R}_{O_{2}}=\boldsymbol{\ell }_{O_{2}},\mathbf{R}^*_{O^{\prime
}}=\mathbf{r}_{O^{\prime }},\mathcal{E}\right \} =\Pr \left \{ \mathbf{R}%
_{I^{\prime }}\leq \boldsymbol{\ell }_{I^{\prime }}\mid \mathbf{R}%
^*_{O^{\prime }}=\mathbf{r}_{O^{\prime }},\mathcal{E}\right \} \text{.}
\end{eqnarray}%
Next, using Equation (\ref{independence}) with $I^{\prime }=\emptyset $, we
also obtain that
\begin{equation}  \label{eq;ind}
\Pr \left \{ \mathbf{R}_{O_{1}}=\boldsymbol{\ell }_{O_{1}}\mid \mathbf{R}%
^*_{O^{\prime }}=\mathbf{r}_{O^{\prime }},\mathcal{E}\right \} =\Pr \left \{
\mathbf{R}_{O_{1}}=\boldsymbol{\ell }_{O_{1}}\mid \mathbf{R}_{O_{2}}=%
\boldsymbol{\ell }_{O_{2}},\mathbf{R}^*_{O^{\prime }}=\mathbf{r}_{O^{\prime
}},\mathcal{E}\right \}.
\end{equation}

Now, pick an arbitrary $o\in O_{1}$. We have
\begin{eqnarray*}
&&\Pr \left \{ \mathbf{R}_{O_{1}}=\boldsymbol{\ell }_{O_{1}}\mid \mathbf{R}%
^*_{O^{\prime }}=\mathbf{r}_{O^{\prime }},\mathcal{E}\right \} \\
&=&\Pr \left \{ \mathbf{R}_{O_{1}\backslash \{o\}}=\boldsymbol{\ell }%
_{O_{1}\backslash \{o\}}\mid R_{o}=\ell _{o},\mathbf{R}^*_{O^{\prime }}=%
\mathbf{r}_{O^{\prime }},\mathcal{E}\right \} \Pr \left \{ R_{o}=\ell
_{o}\mid \mathbf{R}^*_{O^{\prime }}=\mathbf{r}_{O^{\prime }},\mathcal{E}%
\right \} \\
&=&\Pr \left \{ \mathbf{R}_{O_{1}\backslash \{o\}}=\boldsymbol{\ell }%
_{O_{1}\backslash \{o\}}\mid \mathbf{R}^*_{O^{\prime }}=\mathbf{r}%
_{O^{\prime }},\mathcal{E}\right \} \Pr \left \{ R_{o}=\ell _{o}\mid \mathbf{%
R}^*_{O^{\prime }}=\mathbf{r}_{O^{\prime }},\mathcal{E}\right \},
\end{eqnarray*}%
where the last equality comes from Equation (\ref{eq;ind}) above. Now,
applying the argument inductively, we obtain

\begin{equation*}
\Pr \left \{ \mathbf{R}_{O_{1}}=\boldsymbol{\ell }_{O_{1}}\mid \mathbf{R}%
^*_{O^{\prime }}=\mathbf{r}_{O^{\prime }},\mathcal{E}\right \} =\prod_{o\in
O_{1}}\Pr \left \{ R_{o}=\ell _{o}\mid \mathbf{R}^*_{O^{\prime }}=\mathbf{r}%
_{O^{\prime }},\mathcal{E}\right \} \text{.}
\end{equation*}
In other words, conditional on $\mathbf{R}^*_{O^{\prime }}=\mathbf{r}%
_{O^{\prime }}\ $and $\mathcal{E}$, $\{R_{o}\}_{o\in O^{\prime }}$ is a
collection of mutually independent random variables (not necessarily
identically distributed). In addition, conditional on $\{ \mathbf{R}%
^*_{O^{\prime }}=\mathbf{r}_{O^{\prime }}\} \ $and $\mathcal{E}$, for each $%
o\in O^{\prime }$, $R_{o}$\ is stochastically dominated by the uniform
distribution over $\{ \left \lceil \log ^{1+\varepsilon }(n)\right \rceil
+1,...,n\}$. Indeed, the above lemma implies that for any $o\in O^{\prime }$%
, $\Pr \left \{ R_{o}=\ell _{o}\mid \mathbf{R}^*_{O^{\prime }}=\mathbf{r}%
_{O^{\prime }},\mathcal{E}\right \} =0$ if $\ell _{o}\leq r_{o}\ $and is
constant over all possible $\ell _{o}$ such that $\ell _{o}>r_{o}$. Thus, in
the latter case, $\Pr \left \{ R_{o}=\ell _{o}\mid \mathbf{R}^*_{O^{\prime
}}=\mathbf{r}_{O^{\prime }},\mathcal{E}\right \} =\frac{1}{n-r_{o}}$. In
other words, given $\{ \mathbf{R}^*_{O^{\prime }}=\mathbf{r}_{O^{\prime }}\}
\ $and $\mathcal{E}$, for $o\in O^{\prime }$, $R_{o}$ follows a uniform
distribution over $\{r_{o}+1,...,n\}$. Since $o\in O^{\prime }\subseteq \bar{%
O}\subseteq \tilde{O}$, we must have $r_{o}<\log ^{1+\varepsilon }(n)$ and
so $R_{o}$ is stochastically dominated by the uniform distribution over $\{
\left \lceil \log ^{1+\varepsilon }(n)\right \rceil +1,...,n\}$. To recap,
conditional on $\{ \mathbf{R}^*_{O^{\prime }}=\mathbf{r}_{O^{\prime }}\}$
and $\mathcal{E}$,\ $\{R_{o}\}_{o\in O^{\prime }}$\ is a collection of
independent random variables that is stochastically dominated by a
collection of $\left \vert O^{\prime }\right \vert $ iid random variables
distributed according to a uniform distribution over $\{ \left \lceil \log
^{1+\varepsilon }(n)\right \rceil +1,...,n\}$, i.e.,
\begin{eqnarray}
\Pr \left \{ \mathbf{R}_{O^{\prime }}\leq \boldsymbol{\ell }_{O^{\prime
}}\mid \mathbf{R}^*_{O^{\prime }}=\mathbf{r}_{O^{\prime }},\mathcal{E}\right
\} &=&\prod_{o\in O^{\prime }}\Pr \left \{ R_{o}\leq \ell _{o}\mid \mathbf{R}%
^*_{O^{\prime }}=\mathbf{r}_{O^{\prime }},\mathcal{E}\right \}  \label{Ineq}
\\
&\geq &\prod_{o\in O^{\prime }}\Pr \left \{ Y_{o}\leq \ell _{o}\right \}
\text{.}  \notag
\end{eqnarray}%
Now, for any $\boldsymbol{\ell }_{O^{\prime }},\boldsymbol{\ell }_{I^{\prime
}},$%
\begin{eqnarray*}
&&\Pr \left \{ \mathbf{R}_{O^{\prime }}\leq \boldsymbol{\ell }_{O^{\prime }},%
\mathbf{R}_{I^{\prime }}\leq \boldsymbol{\ell }_{I^{\prime }}\mid \mathbf{R}%
^*_{O^{\prime }}=\mathbf{r}_{O^{\prime }},\mathcal{E}\right \} \\
&=&\Pr \left \{ \mathbf{R}_{O^{\prime }}\leq \boldsymbol{\ell }_{O^{\prime
}}\mid \mathbf{R}^*_{O^{\prime }}=\mathbf{r}_{O^{\prime }},\mathcal{E}\right
\} \Pr \left \{ \mathbf{R}_{I^{\prime }}\leq \boldsymbol{\ell }_{I^{\prime
}}\mid \mathbf{R}_{O^{\prime }}\leq \boldsymbol{\ell }_{O^{\prime }},\mathbf{%
R}^*_{O^{\prime }}=\mathbf{r}_{O^{\prime }},\mathcal{E}\right \} \\
&\geq &\prod_{o\in O^{\prime }}\Pr \left \{ Y_{o}\leq \ell _{o}\right \} \Pr
\left \{ \mathbf{R}_{I^{\prime }}\leq \boldsymbol{\ell }_{I^{\prime }}\mid
\mathbf{R}^*_{O^{\prime }}=\mathbf{r}_{O^{\prime }},\mathcal{E}\right \}
\text{.}
\end{eqnarray*}%
where the inequality comes from Equation (\ref{independence_bis}) together
with Equation (\ref{Ineq}). Hence, we obtain
\begin{eqnarray*}
&&\Pr \left \{ \mathbf{R}_{O^{\prime }}\leq \boldsymbol{\ell }_{O^{\prime }},%
\mathbf{R}_{I^{\prime }}\leq \boldsymbol{\ell }_{I^{\prime }}\mid \mathcal{E}%
\right \} \\
&=&\sum_{\mathbf{r}_{O^{\prime }}}\Pr \{ \mathbf{R}^*_{O^{\prime }}=\mathbf{r%
}_{O^{\prime }}\mid \mathcal{E}\} \Pr \left \{ \mathbf{R}_{O^{\prime }}\leq
\boldsymbol{\ell }_{O^{\prime }},\mathbf{R}_{I^{\prime }}\leq \boldsymbol{%
\ell }_{I^{\prime }}\mid \mathbf{R}^*_{O^{\prime }}=\mathbf{r}_{O^{\prime }},%
\mathcal{E}\right \} \\
&\geq &\sum_{\mathbf{r}_{O^{\prime }}}\Pr \{ \mathbf{R}^*_{O^{\prime }}=%
\mathbf{r}_{O^{\prime }}\mid \mathcal{E}\} \prod_{o\in O^{\prime }}\Pr \left
\{ Y_{o}\leq \ell _{o}\right \} \Pr \left \{ \mathbf{R}_{I^{\prime }}\leq
\boldsymbol{\ell }_{I^{\prime }}\mid \mathbf{R}^*_{O^{\prime }}=\mathbf{r}%
_{O^{\prime }},\mathcal{E}\right \} \\
&=&\prod_{o\in O^{\prime }}\Pr \left \{ Y_{o}\leq \ell _{o}\right \} \sum_{%
\mathbf{r}_{O^{\prime }}}\Pr \{ \mathbf{R}^*_{O^{\prime }}=\mathbf{r}%
_{O^{\prime }}\mid \mathcal{E}\} \Pr \left \{ \mathbf{R}_{I^{\prime }}\leq
\boldsymbol{\ell }_{I^{\prime }}\mid \mathbf{R}^*_{O^{\prime }}=\mathbf{r}%
_{O^{\prime }},\mathcal{E}\right \} \\
&=&\prod_{o\in O^{\prime }}\Pr \left \{ Y_{o}\leq \ell _{o}\right \} \Pr
\left \{ \mathbf{R}_{I^{\prime }}\leq \boldsymbol{\ell }_{I^{\prime }}\mid
\mathcal{E}\right \}
\end{eqnarray*}%
as claimed.

Note further that, conditional on $\{ \mathbf{R}^*_{O^{\prime }}=\mathbf{r}%
_{O^{\prime }}\} \ $and $\mathcal{E}$, $\{R_{o}\}_{o\in O^{\prime }} $
stochastically dominates the collection of $\left \vert O^{\prime
}\right
\vert $ iid random variables $U_{1},...,U_{\left \vert O^{\prime
}\right
\vert }$ where the distribution of $U_{o}$\ is uniform over $%
\{1,...,n\}$.\ Using a similar argument as above, we obtain that,
conditional on $\mathcal{E}$, $\{R_{o}\}_{o\in O^{\prime }}$ stochastically
dominates a collection of $\left \vert O^{\prime }\right \vert $ iid random
variables distributed according to a uniform distribution over $\{1,...,n\}$
and we can easily complete the proof of the second part of  \Cref{corner stone}.

Finally, to complete the proof of part (a) of  \Cref{thm: irrelevance}, fix $K$, $x\in[0,1]^K$, and a sequence $y^n\in[0,1]^n$. Let $O_K=\{o_1,\ldots,o_K\}$ and write $A_n=\{O_K\subseteq \bar O\}$. By \Cref{prop:short-cycles}, $\Pr(A_n)\to1$. Conditional on $A_n$, \Cref{corner stone} squeezes the CDF of $(\bar R_{o_1},\ldots,\bar R_{o_K},\bar R_{i_1},\ldots,\bar R_{i_n})$ between the corresponding CDFs obtained by replacing the object ranks by iid uniforms on $\{1,\ldots,n\}$ and by iid uniforms on $\{\lceil\log^{1+\varepsilon}n\rceil+1,\ldots,n\}$. These two product-uniform bounds differ by at most $K\lceil\log^{1+\varepsilon}n\rceil/n=o(1)$, uniformly in $x$ and $y^n$. Since conditioning on $A_n$ changes any probability by at most $\Pr(A_n^c)=o(1)$, the desired growing-vector CDF convergence follows. 

\subsection{Proof of  \Cref{thm: irrelevance} (b)}

\label{sec: empirical dist}

Fix $x\in \lbrack 0,1]$. By the above result, given $\{ \bar{O}=O^{\prime
\prime }\}$, the collection $\{ \mathbf{1}_{\{ \bar{R}_{o}\leq x\}} \}_{o\in
\bar{O}}$ is stochastically dominated by $\{ \mathbf{1}_{\{ \bar{Y}_{o}\leq
x\}} \}_{o\in \bar{O}}$,  where $\bar{Y}_{o}$ is $\frac{1}{n}U\{ \left
\lceil
\log ^{1+\varepsilon }(n)\right \rceil +1,...,n\}$ which converges in
distribution to $U[0,1]$. Similarly, given $\{ \bar{O}=O^{\prime \prime }\}$%
, the collection $\{ \mathbf{1}_{\{\bar{R}_{o}\leq x\}}\}_{o\in \bar{O}}$
stochastically dominates the collection $\{ \mathbf{1}_{\{ \bar{U}_{o}\leq
x\}} \}_{o\in \bar{O}}$ where $\bar{U}_{o}$ is $\frac{1}{n}U\{1,...,n\}$
which converges in distribution to $U[0,1]$.

Now fix any $\delta>0$ and condition on the event $|\bar O|\ge(1-\delta)n$, whose probability tends to one, and on the realized set $\bar O=O''$. On this event, the contribution of $O\setminus\bar O$ to the empirical average is at most $\delta$. Hence
\[
\frac{|\bar O|}{n}\cdot \frac1{|\bar O|}\sum_{o\in\bar O}\mathbf 1_{\{\bar R_o\le x\}}
\le \frac1n\sum_{o\in O}\mathbf 1_{\{\bar R_o\le x\}}
\le
\frac{|\bar O|}{n}\cdot \frac1{|\bar O|}\sum_{o\in\bar O}\mathbf 1_{\{\bar R_o\le x\}}+\delta .
\]
The stochastic squeeze from \Cref{corner stone} and the law of large numbers for the two iid uniform bounds imply that the middle empirical average lies between $(1-\delta)x-o_p(1)$ and $x+\delta+o_p(1)$, uniformly over the conditioning on $O''$. Removing the high-probability conditioning and then letting $\delta\downarrow0$ gives
\[
\frac1n\sum_{o\in O}\mathbf 1_{\{\bar R_o\le x\}}\overset{p}{\longrightarrow}x .
\]

\section{Proof of Proposition~\ref{prop:je-metrics} and Corollary~\ref{cor: BP}}
\label{sec: BP}

\begin{proof}[Proof of Proposition~\ref{prop:je-metrics}]
For a mechanism $M\in\{\mathrm{TTC},\mathrm{RSD}\}$, let $\mu^M$ be the matching it produces.  For each agent-object pair $(i,o)$, define
\[
        E_{io}^M:=\mathbf 1\{oP_i\mu^M(i)\}
\]
and
\[
        J_{io}^M:=E_{io}^M\mathbf 1\{i\succ_o(\mu^M)^{-1}(o)\}.
\]
Thus, $E_{io}^M$ records an envy incidence, while $J_{io}^M$ records a justified-envy incidence, equivalently a blocking pair involving agent $i$ and object $o$.  Let
\[
        E_n^M:=\mathbb E\left[\sum_{i,o}E_{io}^M\right],
        \qquad
        J_n^M:=\mathbb E\left[\sum_{i,o}J_{io}^M\right].
\]
The first statistic in the proposition is $J_n^M/E_n^M$.  The third statistic is $J_n^M/[n(n-1)]$.

We first compute the expected number of envy incidences.  Let $R_i^M$ be the preference rank of agent $i$'s assignment under mechanism $M$, where rank one is best.  Since agent $i$ envies exactly the objects ranked above her assignment,
\begin{equation}
        \sum_o E_{io}^M=R_i^M-1.
        \label{eq:je-envy-rank-count}
\end{equation}
Under RSD, conditional on agent $i$'s serial position being $s$, there are $m=n-s+1$ objects available when she chooses.  These objects form, from $i$'s perspective, a uniformly random $m$-element subset of $\{1,\ldots,n\}$, and she receives the best-ranked object in that subset.  If $Y$ denotes the minimum of a uniformly random $m$-element subset of $\{1,\ldots,n\}$, then
\[
        \Pr\{Y>q\}=\frac{\binom{n-q}{m}}{\binom nm},
        \qquad q=0,\ldots,n-m.
\]
Thus the tail-sum formula and the hockey-stick identity give
\[
        \mathbb E[Y]
        =
        \sum_{q=0}^{n-m}\frac{\binom{n-q}{m}}{\binom nm}
        =
        \frac{\sum_{r=m}^{n}\binom rm}{\binom nm}
        =
        \frac{\binom{n+1}{m+1}}{\binom nm}
        =
        \frac{n+1}{m+1}.
\]
Averaging over the uniform serial position $s$ yields
\begin{equation}
        \mathbb E[R_i^{\mathrm{RSD}}]
        =
        \frac1n\sum_{s=1}^n\frac{n+1}{n-s+2}
        =
        \frac{n+1}{n}(H_{n+1}-1).
        \label{eq:je-rsd-expected-rank}
\end{equation}
By the finite agent-side equivalence between random-priority TTC and RSD \citep{pathak/sethuraman:11}, $R_i^{\mathrm{TTC}}$ and $R_i^{\mathrm{RSD}}$ have the same distribution.  Combining this equivalence with (\ref{eq:je-envy-rank-count}) and (\ref{eq:je-rsd-expected-rank}),
\begin{equation}
        E_n^{\mathrm{TTC}}=E_n^{\mathrm{RSD}}
        =
        (n+1)(H_{n+1}-1)-n
        \sim n\log n.
        \label{eq:je-total-envy}
\end{equation}

Under RSD, priorities are independent of the assignment.  Conditional on any envy incidence, the envying agent and the recipient of the envied object are two distinct agents and are symmetrically ordered in that object's priority ranking.  Hence
\begin{equation}
        J_n^{\mathrm{RSD}}=\frac12E_n^{\mathrm{RSD}}.
        \label{eq:je-rsd-half}
\end{equation}
This proves the RSD part of the incidence-ratio statement.  Together with (\ref{eq:je-total-envy}), it also gives the stated RSD blocking-pair fraction.

We next prove the TTC incidence ratio.  For each object $o$, let $L_o=1$ if $o$ is assigned through a long cycle under TTC, and let $L_o=0$ otherwise.  If $o$ is assigned through a short cycle to agent $k$, then $o$ points to $k$ in the round in which it is assigned.  Any agent who envies $k$ for $o$ must still be present in that round; otherwise, $o$ would have been available when that agent received her own assignment.  Since $k$ is the highest-priority remaining agent for $o$ in that round, no such envy can be justified.  Therefore,
\begin{equation}
        J_{io}^{\mathrm{TTC}}(1-L_o)=0
        \quad\text{for all }(i,o).
        \label{eq:je-short-cycle-none}
\end{equation}
For long-cycle objects, conditional on an envy incidence, the envy is justified with
probability one half. Fix $(i,o)$ and consider a profile for which
$E_{io}^{\mathrm{TTC}}L_o=1$. Let $k=(\mu^{\mathrm{TTC}})^{-1}(o)$ be the
recipient of $o$, and let $a_o$ be the agent to whom $o$ points in the round in
which $o$ is assigned. Since the cycle is long, $a_o\ne k$. Since $i$ envies the
recipient of $o$, agent $i$ must still be present in that round; otherwise $o$
would have been available when $i$ was assigned, contradicting $oP_i\mu^{\mathrm{TTC}}(i)$.
Hence both $i$ and $k$ have lower priority for $o$ than $a_o$. Now exchange only
$i$ and $k$ in object $o$'s priority ordering, leaving all preferences and all
other priority orderings unchanged. In every round before $o$ is assigned, neither
$i$ nor $k$ can be the highest-priority remaining agent for $o$, since $a_o$ is
still present and has higher priority than both; after $o$ is assigned, $o$'s
priority ordering is irrelevant. Thus the TTC path, the recipient $k$ of $o$, and
the event $E_{io}^{\mathrm{TTC}}L_o=1$ are unchanged. The exchange only reverses the
relative priority of $i$ and $k$ at $o$, and therefore reverses whether the envy is
justified. Since the priority ordering at $o$ is drawn uniformly, the original
profile and the exchanged profile have the same probability. So among long-cycle
envy incidences, exactly half are justified and half are unjustified. Summing over
all $(i,o)$ gives
\begin{equation}
        \mathbb E\left[\sum_{i,o}J_{io}^{\mathrm{TTC}}L_o\right]
        =
        \frac12\mathbb E\left[\sum_{i,o}E_{io}^{\mathrm{TTC}}L_o\right].
        \label{eq:je-long-cycle-half}
\end{equation}

Combining (\ref{eq:je-short-cycle-none}) and (\ref{eq:je-long-cycle-half}),
\begin{equation}
        J_n^{\mathrm{TTC}}
        =
        \frac12\left\{
        E_n^{\mathrm{TTC}}
        -
        \mathbb E\left[\sum_{i,o}E_{io}^{\mathrm{TTC}}(1-L_o)\right]
        \right\}.
        \label{eq:je-ttc-decomposition}
\end{equation}
It remains to show that the short-cycle term in (\ref{eq:je-ttc-decomposition}) is $o(n\log n)$.

Let $C_n:=\{o:L_o=0\}$ be the set of short-cycle objects, and let $X_o:=\sum_iE_{io}^{\mathrm{TTC}}$ be the number of envy incidences directed at the recipient of object $o$.  Then
\[
        \sum_{i,o}E_{io}^{\mathrm{TTC}}(1-L_o)=\sum_{o\in C_n}X_o.
\]
By Cauchy's inequality across objects and then in the probability space,
\begin{equation}
        \mathbb E\left[\sum_{o\in C_n}X_o\right]
        \le
        \left(\mathbb E|C_n|\right)^{1/2}
        \left(\mathbb E\sum_oX_o^2\right)^{1/2}.
        \label{eq:je-two-cauchy}
\end{equation}
Indeed, pointwise, $\sum_{o\in C_n}X_o\le |C_n|^{1/2}(\sum_oX_o^2)^{1/2}$; applying Cauchy's inequality to the expectation of this product yields (\ref{eq:je-two-cauchy}).

We bound the two factors in (\ref{eq:je-two-cauchy}).  Let $S_t$ be the number of short-cycle objects assigned in round $t$, with $S_t=0$ after TTC terminates, and let $\mathcal F_{t-1}$ denote the history at the beginning of round $t$.  Since $\{T\ge t\}$ is $\mathcal F_{t-1}$-measurable and Proposition~\ref{prop: conditional exp short cycle} gives $\mathbb E[S_t\mid\mathcal F_{t-1}]\le2$,
\[
        \mathbb E[S_t]
        =
        \mathbb E\left[\mathbf 1\{T\ge t\}\mathbb E[S_t\mid\mathcal F_{t-1}]\right]
        \le
        2\Pr\{T\ge t\}.
\]
Therefore $\mathbb E|C_n|\le2\mathbb E[T]$.  Since $T\le n$ and Proposition~\ref{prop:stopping} gives $T/n\to0$ in probability, bounded convergence gives $\mathbb E[T]/n\to0$.  Hence
\begin{equation}
        \mathbb E|C_n|=o(n).
        \label{eq:je-short-count-small}
\end{equation}
For the second factor, use the finite equivalence between random-priority TTC and RSD at the level of the assignment distribution, conditional on preferences.  The statistic $\sum_oX_o^2$ depends only on preferences and on the realized matching, so it is enough to bound it under RSD.  Fix an object $o$, and let $\tau$ be the serial position at which $o$ is chosen.  Agents before $\tau$, and the agent at $\tau$, do not envy the recipient of $o$.  For $t>\tau$, let $Z_t$ be the indicator that the agent in serial position $t$ envies the recipient of $o$.  Conditional on the RSD history before position $t$, this occurs with probability $p_t=1/(n-t+2)$.  Iterated expectations give $\mathbb E[Z_sZ_t\mid \tau]=p_sp_t$ for $\tau<s<t$.  Therefore,
\[
        \mathbb E[X_o^2\mid \tau]
        \le
        \sum_{t=\tau+1}^np_t+
        \left(\sum_{t=\tau+1}^np_t\right)^2
        \le
        H_n+H_n^2.
\]
By symmetry across objects,
\begin{equation}
        \mathbb E\sum_oX_o^2=O(n\log^2n).
        \label{eq:je-X-second-moment}
\end{equation}
Equations (\ref{eq:je-two-cauchy}), (\ref{eq:je-short-count-small}), and (\ref{eq:je-X-second-moment}) imply
\begin{equation}
        \mathbb E\left[\sum_{i,o}E_{io}^{\mathrm{TTC}}(1-L_o)\right]
        =o(n\log n).
        \label{eq:je-short-envy-negligible}
\end{equation}
Together with (\ref{eq:je-ttc-decomposition}) and (\ref{eq:je-total-envy}), this proves that $J_n^{\mathrm{TTC}}/E_n^{\mathrm{TTC}}\to1/2$.  The TTC blocking-pair fraction is $J_n^{\mathrm{TTC}}/[n(n-1)]$, so (\ref{eq:je-total-envy}) and (\ref{eq:je-short-envy-negligible}) also imply that it is asymptotic to $\log n/(2n)$.

It remains to prove the agent-level statement.  Let
\[
        A_n^M:=\mathbb E\left[\frac1n\sum_i\mathbf 1\left\{\sum_oJ_{io}^M\ge1\right\}\right]
\]
be the expected fraction of agents with at least one justified envy.  We first compute this number under RSD.  For a fixed agent,
\begin{equation}
        \Pr\{R_i^{\mathrm{RSD}}>r\}
        =
        \frac1n\sum_{s=1}^n\frac{\binom{s-1}{r}}{\binom nr}
        =
        \frac{n-r}{n(r+1)},
        \qquad r=0,\ldots,n.
        \label{eq:je-rsd-rank-tail}
\end{equation}
Thus
\begin{equation}
        \Pr\{R_i^{\mathrm{RSD}}=r\}
        =
        \frac{n+1}{nr(r+1)},
        \qquad r=1,\ldots,n.
        \label{eq:je-rsd-rank-mass}
\end{equation}
Conditional on $R_i^{\mathrm{RSD}}=r$, agent $i$ envies exactly $r-1$ recipients.  Priorities are independent across objects and independent of the RSD assignment, so the probability that none of these $r-1$ envies is justified is $2^{-(r-1)}$.  Hence
\begin{equation}
        A_n^{\mathrm{RSD}}
        =
        1-
        \sum_{r=1}^n
        2^{-(r-1)}\frac{n+1}{nr(r+1)}.
        \label{eq:je-rsd-agent-exact}
\end{equation}
Letting $n\to\infty$,
\begin{equation}
        \lim_{n\to\infty}A_n^{\mathrm{RSD}}
        =
        1-
        \sum_{r=1}^{\infty}\frac{2^{-(r-1)}}{r(r+1)}
        =
        2\log2-1,
        \label{eq:je-rsd-agent-limit}
\end{equation}
where the last equality follows from
\[
        \sum_{r=1}^{\infty}\frac{2^{-(r-1)}}{r(r+1)}
        =
        \int_0^1\frac{1-x}{1-x/2}\,dx
        =
        2(1-\log2).
\]

We now transfer this calculation to TTC.  Fix $K$.  Consider envy incidences generated by agents whose assignment rank is at most $K$ and directed at short-cycle objects.  If such an agent envies object $o$, then $o$ must be among her top $K$ objects.  Let $N_o^K$ be the number of agents who rank $o$ among their top $K$ objects.  The expected number of these finite-rank envy incidences directed at short-cycle objects is at most $\mathbb E[\sum_{o\in C_n}N_o^K]$.  By the same two Cauchy inequalities used above,
\[
        \mathbb E\left[\sum_{o\in C_n}N_o^K\right]
        \le
        \left(\mathbb E|C_n|\right)^{1/2}
        \left(\mathbb E\sum_o(N_o^K)^2\right)^{1/2}.
\]
For each fixed object, $N_o^K$ is binomial with parameters $n$ and $K/n$, so $\mathbb E[(N_o^K)^2]\le K+K^2$.  Using (\ref{eq:je-short-count-small}), we obtain
\begin{equation}
        \frac1n\mathbb E\left[\sum_{o\in C_n}N_o^K\right]\longrightarrow0
        \quad\text{for every fixed }K.
        \label{eq:je-finite-rank-short-negligible}
\end{equation}
Thus, for a fixed agent, the probability that she both has assignment rank at most $K$ and envies some short-cycle object goes to zero.

Now fix $r\le K$. Let
\[
        \mathcal V_i:=\{o\in O: oP_i\mu^{\mathrm{TTC}}(i)\}
\]
be the set of objects envied by $i$, so that $|\mathcal V_i|=r-1$ on
$\{R_i^{\mathrm{TTC}}=r\}$. Let $A_i$ be the event that every object in
$\mathcal V_i$ is assigned through a long cycle. By \Cref{eq:je-finite-rank-short-negligible},
$\Pr\{R_i^{\mathrm{TTC}}=r,\ A_i^c\}=o(1)$ for each fixed $r\le K$.

Conditional on $R_i^{\mathrm{TTC}}=r$ and $A_i$, fix a profile and write, for
each $o\in\mathcal V_i$, $k_o=(\mu^{\mathrm{TTC}})^{-1}(o)$ for the recipient of
$o$ and $a_o$ for the agent to whom $o$ points in the round in which $o$ is
assigned. Since $o$ is assigned through a long cycle, $a_o\ne k_o$. Moreover,
agent $i$ is still present in that round, and she cannot be assigned in that
round; otherwise, with $o$ still available, she would receive her favorite
remaining object and could not envy the recipient of $o$. Thus
$a_o\notin\{i,k_o\}$, and both $i$ and $k_o$ have lower priority for $o$ than
$a_o$.

Now take any subset $S\subseteq\mathcal V_i$. For each $o\in S$, exchange only
$i$ and $k_o$ in object $o$'s priority ordering, leaving all preferences and all
other priority orderings unchanged. Along the original TTC path, before such an
object $o$ is assigned, the agent $a_o$ is still present and has higher priority
for $o$ than both $i$ and $k_o$. Hence, under the same sequence of remaining
sets, object $o$'s pointer is unchanged before its assignment; after $o$ is
assigned, its priority ordering is irrelevant. Since preferences and all other
object priorities are unchanged, an induction over rounds shows that the entire
TTC path is unchanged. Consequently, the assignment of $i$, the set
$\mathcal V_i$, and the event $A_i$ are unchanged. For each $o\in S$, however,
the exchange reverses the relevant comparison $i\succ_o k_o$, while the
corresponding comparison is unchanged for each $o\notin S$.

Since object-specific priority orderings are drawn independently and uniformly,
the profiles obtained from the $2^{r-1}$ choices of $S\subseteq\mathcal V_i$ all
have the same probability. Hence all $2^{r-1}$ possible patterns of justified
versus unjustified envy across the $r-1$ envied long-cycle objects are equally
likely. In particular, conditional on $R_i^{\mathrm{TTC}}=r$ and $A_i$, the
probability that none of the $r-1$ envy incidences is justified is
$2^{-(r-1)}$. Combining this with
$\Pr\{R_i^{\mathrm{TTC}}=r,\ A_i^c\}=o(1)$ gives, for each fixed $r$,
\begin{equation}
        \Pr\left\{\sum_oJ_{io}^{\mathrm{TTC}}=0,\ R_i^{\mathrm{TTC}}=r\right\}
        =
        2^{-(r-1)}\Pr\{R_i^{\mathrm{TTC}}=r\}+o(1).
        \label{eq:je-ttc-agent-fixed-r}
\end{equation}

The finite agent-side equivalence between TTC and RSD gives $\Pr\{R_i^{\mathrm{TTC}}=r\}=\Pr\{R_i^{\mathrm{RSD}}=r\}$.  Summing (\ref{eq:je-ttc-agent-fixed-r}) over $r\le K$, using (\ref{eq:je-rsd-rank-mass}), and then letting first $n\to\infty$ and then $K\to\infty$ gives
\[
        \lim_{n\to\infty}\Pr\left\{\sum_oJ_{io}^{\mathrm{TTC}}=0\right\}
        =
        \sum_{r=1}^{\infty}\frac{2^{-(r-1)}}{r(r+1)}
        =
        2(1-\log2),
\]
The omitted terms with $r>K$ are negligible after taking $K$ large: indeed,
\[
\Pr\left\{\sum_o J_{io}^{\mathrm{TTC}}=0,\ R_i^{\mathrm{TTC}}>K\right\}
\le
\Pr\{R_i^{\mathrm{TTC}}>K\}
=
\Pr\{R_i^{\mathrm{RSD}}>K\}
=
\frac{n-K}{n(K+1)}
\to
\frac{1}{K+1},
\]
where the equality in distribution follows from the finite agent-side equivalence between
TTC and RSD. Hence this tail vanishes when we subsequently let $K\to\infty$. Therefore $A_n^{\mathrm{TTC}}\to2\log2-1$.  This completes the proof of the proposition.
\end{proof}

\begin{proof}[Proof of Corollary~\ref{cor: BP}]
By exchangeability across ordered agent-object pairs,
\[
        \Pr\{E_{io}^M=1\}=\frac{E_n^M}{n^2},
        \qquad
        \Pr\{J_{io}^M=1\}=\frac{J_n^M}{n^2}.
\]
The stated envy probability follows from (\ref{eq:je-total-envy}). For RSD, (\ref{eq:je-rsd-half}) implies that the fixed-pair blocking probability is exactly one
half of the fixed-pair envy probability. For TTC, Proposition~\ref{prop:je-metrics} gives
$J_n^{\mathrm{TTC}}/E_n^{\mathrm{TTC}}\to 1/2$, so the fixed-pair blocking
probability is asymptotically one half of the fixed-pair envy probability. Since
$J_{io}^M=1$ implies $E_{io}^M=1$, we have
\[
    \Pr\{J_{io}^M=1\mid E_{io}^M=1\}
    =
    \frac{\Pr\{J_{io}^M=1\}}{\Pr\{E_{io}^M=1\}}.
\]
Thus the conditional probability that the envy is justified is exactly $1/2$ under
RSD and converges to $1/2$ under TTC.
\end{proof}

\bibliographystyle{economet}
\bibliography{bibmatching}


\clearpage
\setcounter{page}{1}
\renewcommand{\thepage}{OA-\arabic{page}}

\section*{Online Appendix: Not for Publication}
\phantomsection
\addcontentsline{toc}{section}{Online Appendix}

\setcounter{section}{0}
\renewcommand{\thesection}{OA.\arabic{section}}
\renewcommand{\thesubsection}{\thesection.\arabic{subsection}}

\renewcommand{\theHsection}{OA.\arabic{section}}
\renewcommand{\theHsubsection}{OA.\arabic{section}.\arabic{subsection}}

\section{Additional Details for the Proof of \Cref{Markov}}

\label{OA:MarkovProof}

\subsection{Computation of $\beta$}\label{sec:lem:beta}

To compute $\beta$, we first establish the following lemma.

\subsubsection{Useful computational lemma}


For the next result, consider agents $I^{\prime }$ and objects $O^{\prime }$
such that $|I^{\prime }|=|O^{\prime }|=m>0$. We say a mapping $f=h\circ g$
is a \textbf{bipartite bijection}, if $g:I^{\prime }\to O^{\prime }$ and $%
h:O^{\prime }\to I^{\prime }$ are both bijections. A \textbf{cycle} of a
bipartite bijection is a cycle of the induced digraph. Note that a bipartite
bijection consists of disjoint cycles. A \textbf{random bipartite bijection}
is a (uniform) random selection of a bipartite bijection from the set of all
bipartite bijections. The following result will prove useful for a later
analysis.

\begin{lemma}
\label{lem:lovasz} \label{trick} Fix sets $I^{\prime }$ and $O^{\prime }$
with $|I^{\prime }|=|O^{\prime }|=m>0$, and a subset $K\subset I^{\prime
}\cup O^{\prime }$, containing $a\ge 0$ vertices in $I^{\prime }$ and $b\ge 0
$ vertices in $O^{\prime }$. The probability that each cycle in a random
bipartite bijection contains at least one vertex from $K$ is
\begin{equation*}
\frac{a+b}{m}- \frac{ab}{m^2}.
\end{equation*}
\end{lemma}

\begin{proof}
We begin with a few definitions.  A \textbf{permutation} of $X$ is a bijection $f:X\rightarrow X$.  A \textbf{cycle} of a permutation is a cycle of the digraph induced by the permutation.  A permutation consists of disjoint cycles.  A \textbf{random
		permutation} chooses  uniformly at random a permutation $f$ from the set of all possible
	permutations.  Our proof will invoke following result:
	\begin{fact} [\citet{lovasz;79} Exercise 3.6] \label{lovasz}
		The  probability that each cycle of a random permutation of a finite set $X$ contains at least one element of a set $%
		Y\subset X$ is $|Y|/|X|$.
	\end{fact}
	
	To begin, observe first that a bipartite bijection $h\circ g$ induces a
	permutation of set $I^{\prime }$. Thus, a random bipartite bijection defined by a pair $(g,h)$ over $I'\times O'$ induces
	a random permutation of  $I^{\prime }$. Moreover, because $h$ is drawn uniformly and independently of $g$, the induced uniform permutation $h\circ g$ is independent of $g$. Consequently, it is independent of the random set of vertices of $I'$ whose $g$-image lies in $K\cap O'$. To compute the probability that each cycle of a random bipartite bijection  $h\circ g$ contains at least one vertex in $K\subset I'\cup O'$, we shall apply Fact \ref{lovasz} to this induced random permutation of $I'$.

	Indeed, each cycle of a random bipartite bijection contains at least one vertex in
	$K\subset I'\cup O'$ if and only if each cycle of the induced random permutation of $I'$ contains either a vertex in $K \cap I'$ or a vertex in $ I'\setminus K$ that  points to a vertex in $K\cap O'$  in the original random bipartite bijection.  Hence, the relevant set $Y\subset I'$ for the purpose of applying Fact \ref{lovasz} is a random set that contains  $|K\cap I'|=a$ vertices of the former kind and $Z$ vertices of the latter kind.
	
	The number $Z$ is random and takes a value $z$, $\max \{b-a,0\}\leq z\leq \min \{m-a,b\}$,  with probability:
	\begin{equation*}
	\Pr \{Z=z\}=\frac{{\binom{m-a}{z}}{\binom{a}{b-z}}}{{\binom{m}{b}}}.
	\end{equation*}%
	This formula is explained as follows. $\Pr \{Z=z\}$ is the ratio of the
	number of bipartite bijections having exactly $z$ vertices in $I^{\prime }\setminus K$
	pointing toward $K\cap O^{\prime }$ to the total number of bipartite bijections.
	
	Note that since we consider bipartite bijections, the number of vertices in $I^{\prime}$ pointing
	to the vertices in $K\cap O^{\prime }$ must be equal to $b$.
	Focusing first on the
	numerator, we have to compute the number of bipartite bijections having exactly $z$ vertices in $I^{\prime }\setminus K$
	pointing toward $K\cap O^{\prime }$ and the remaining $b-z$ vertices pointing to the remaining $K\cap O^{\prime }$.   There are
	${\binom{m-a}{z%
		}\binom{a}{b-z}}$ ways one can choose  $z$ vertices from $I^{\prime }\setminus K$ and $b-z$ vertices from
	$K\cap I^{\prime }$.    Thus, the total number of
	bipartite bijections having exactly $z$ vertices in $I^{\prime }\setminus K$
	that point to $K\cap O^{\prime }$ is ${\binom{m-a}{z}\binom{a}{b-z}}\upsilon$, where
	$\upsilon$ is the total number of bipartite bijections in which the $b$ vertices thus chosen point to the vertices in $K\cap O^{\prime }$. This gives us the numerator.   As for the denominator,  the total number of
	bipartite bijections having $b$ vertices in $I^{\prime }$ pointing to $K\cap
	O^{\prime }$ is ${\binom{m}{b}}$ (the number of ways $b$ vertices are chosen from $I'$), multiplied by $\upsilon$ (the number of bijections in which the $b$ vertices thus chosen point to the vertices in $K\cap O^{\prime }$). Hence, the denominator is
	${\binom{m}{b}}\upsilon$.   Thus, we get the above formula.
	
	Recall our goal is to compute the probability that each cycle of the random permutation induced by the random bipartite bijection contains at least one vertex in the random set $Y$, with $|Y|=a+Z$, where $\Pr\{Z=z\}=\frac{{\binom{m-a}{z}}{\binom{a}{b-z}}}{{\binom{m}{b}}}$.   Applying Fact \ref{lovasz}, then the desired probability is
	\begin{align*}
	\E\left[ \frac{|Y|}{|I'|}\right]
	=& \sum_{z=\max \{b-a,0\}}^{\min \{m-a,b\}}\Pr \{Z=z\}\frac{a+z}{m} \\
	=& \frac{a}{m}+\sum_{z=\max \{b-a,0\}}^{\min \{m-a,b\}}\Pr \{Z=z\}\frac{z}{m}
	\\
	=& \frac{a}{m}+\sum_{z=\max \{b-a,0\}}^{\min \{m-a,b\}}\frac{{\binom{m-a}{z}}{\binom{a}{b-z}}%
	}{{\binom{m}{b}}}\left( \frac{z}{m}\right) \\
	=& \frac{a}{m}+\left( \frac{m-a}{m{\binom{m}{b}}}\right) \sum_{z=\max
		\{b-a,1\}}^{\min \{m-a,b\}}{\binom{a}{b-z}}{\binom{m-a-1}{z-1}} \\
	=& \frac{a}{m}+\left( \frac{m-a}{m{\binom{m}{b}}}\right) {\binom{m-1}{b-1}}
	\\
	=& \frac{a}{m}+\frac{b(m-a)}{m^{2}} \\
	=& \frac{a+b}{m}-\frac{ab}{m^{2}},
	\end{align*}%
	where the fifth equality follows from Vandermonde's identity.
\end{proof}

\subsection{Completion of the computation of $\beta$}

Recall that, given an arbitrary $F\in \mathcal{F}_{N_{i+1},k_{i+1}}$, $\beta (I_i,O_i,k_i^I, k_i^O ;I_{i+1},O_{i+1},k_{i+1}^I, k_{i+1}^O )$
counts, the number of pairs $(F^{\prime },\phi )$, $F^{\prime }\in \mathcal{F}_{N_{i},k_{i}}$, causing $F$ to arise.

As mentioned in the main text, we can construct all such pairs by choosing a quadruplet $%
	(a,b,c,d)$ of four non-negative integers with $a+c=k_i^I$ and  $b+d=k_i^O$,
	
	\begin{enumerate}
		\item [(i)] choosing $c$ old roots from $I_{i+1}$, and similarly, $d$ old roots
		from $O_{i+1}$,
		
		\item [(ii)] choosing $a$ old roots from $I_{i}\backslash I_{i+1}$ and similarly, $b
		$ old roots from $O_{i}\backslash O_{i+1}$,
		
		\item [(iii)] choosing a partition into cycles of $N_{i}\backslash N_{i+1}$, each
		cycle of which contains at least one old root from (ii),\footnote{Within round $i$ of TTC, one cannot have a cycle creating
			only with nodes that are not roots in the forest obtained at the beginning of round $i$. This is due to the simple fact that a
			forest is an acyclic graph. Thus, each cycle creating must contain at least one old root. Given that, by definition, these roots are eliminated
			from the set of available nodes in round $i+1$, these old roots that each cycle must contain must be from (ii).}
		
		\item [(iv)]  choosing a mapping of the $k_{i+1}^I+k_{i+1}^O$ new roots to $N_{i}\backslash
		N_{i+1}$.
	\end{enumerate}

Hence, setting $n=|I_i|$, $o=|O_i|$ and $m=|I_i|-|I_{i+1}|=|O_i|-|O_{i+1}|$, the number of such pairs is computed as
	\begin{align*}
	& \sum_{b+d=k_i^O } \sum_{a+c=k_i^I }{\binom{n-m}{c}}{\binom{o-m}{d}}{\binom{m}{a}}{%
		\binom{m}{b}}\left( \frac{a+b}{m}-\frac{ab}{m^{2}}\right) (m!)^{2}m^{k_{i+1}^I+k_{i+1}^O
	} \\
	=& (m!)^{2}m^{k_{i+1}^I+k_{i+1}^O }\times \left( \sum_{b+d=k_i^O } \sum_{a+c=k_i^I }{\binom{n-m}{c}}{%
		\binom{o-m}{d}}{\binom{m-1}{a-1}}{\binom{m}{b}}\right.  \\
	& \quad + \sum_{b+d=k_i^O } \sum_{a+c=k_i^I }{\binom{n-m}{c}}{\binom{o-m}{d}}{%
		\binom{m}{a}}{\binom{m-1}{b-1}} \\
	& \quad \left.- \sum_{b+d=k_i^O } \sum_{a+c=k_i^I }{\binom{n-m}{c}}{%
		\binom{o-m}{d}}{\binom{m-1}{a-1}}{\binom{m-1}{b-1}}\right)  \\
		=& (m!)^{2}m^{k_{i+1}^I+k_{i+1}^O }\times \left( {\binom{o}{k_i^O}}{%
			\binom{n-1}{k_i^I-1}}   +  {\binom{n}{k_i^I}}{\binom{o-1}{k_i^O-1}}-{\binom{n-1}{k_i^I-1}}{%
			\binom{o-1}{k_i^O-1}}\right)  \\
	=& \frac{(m!)^{2}m^{k_{i+1}^I+k_{i+1}^O}}{no}{\binom{n}{k_i^I}}{\binom{o}{k_i^O}}(nk_i^O+ok_i^I-k_i^Ik_i^O).
	\end{align*}%
	The first equality follows from  \Cref{trick}, along with the fact that
	there are $(m!)^{2}$ possible bipartite bijections between $n-m
	$ agents and $o-m$ objects, and the fact that there are $%
	m^{k_{i+1}^I}m^{k_{i+1}^{O}}$ ways in which new roots $%
	k_{i+1}^I$ agents and $k_{i+1}^O$ objects could have
	pointed to $2m$ cyclic vertices ($m$ on the individuals' side and $m$ on the
	objects' side), the third equality follows from Vandermonde's identity, and the last equality follows from simplifying terms.



\subsection{Computation of transition probability}\label{sec:computation transition}

Before going through the algebra, we need the following lemma characterizing the number of spanning rooted forests.

\begin{lemma}[\citet{jin-liu:04}]
\label{forest} 
For fixed specified root sets with $k$ roots in $I$ and $\ell$ roots in $O$, the number of spanning rooted forests is $f(n,o,k,\ell):=o^{n-k-1}n^{o-\ell-1}(\ell n+ko-k\ell)$. If only the numbers of roots are specified, this expression is multiplied by $\binom nk\binom o\ell$.
\end{lemma}

Now, we have
\begin{align*}
	\Upsilon =&  o^{k^I}n^{k^O}{\binom{n}{{k^I}}}{\binom{o}{%
			k^O}}f(n,o,k^I,k^O) \\
	=&  o^{k^I}n^{k^O}{\binom{n}{{k^I}}}{\binom{o}{%
			k^O}}o^{n-k^I-1}n^{o-k^O-1}(nk^O+ok^I-k^Ik^O) \\
	=&  {\binom{n}{k^I}}{\binom{o}{k^O}}%
	o^{n-1}n^{o-1}(nk^O+ok^I-k^Ik^O).
	\end{align*}%
where the second equality follows from  \Cref{forest}.

$\T$ is now computed as:
	\begin{align*}
	\Theta & ={\binom{n}{m}}{\binom{o}{m}}f(n-m,o-m,\lambda ^I,\lambda
	^O)\beta (n,o,k^I, k^O ;m,\lambda ^I,\lambda ^O) \\
	& ={\binom{n}{m}}{\binom{o}{m}}f(n-m,o-m,\lambda ^I,\lambda
	^O)\frac{(m!)^{2}m^{\lambda^I+\lambda^O}}{no}{\binom{n}{k^I}}{\binom{o}{k^O}}(nk^O+ok^I-k^Ik^O)\\
	& =f(n-m,o-m,\lambda ^I,\lambda ^O)\left( \frac{n!}{(n-m)!}\right) \left( \frac{o!}{(o-m)!}\right)
	\frac{m^{\lambda ^I+\lambda ^O}}{no}{\binom{n}{k^I}}{\binom{o}{k^O}}(nk^O+ok^I-k^Ik^O).
	\end{align*}%

Collecting the terms we obtain
	\begin{equation*}
	\mathbf{P}(n,o,k^I,k^O ;m,\lambda ^I,\lambda ^O)=\frac{\Theta}{\Upsilon}
	=\frac{1}{o^{n}n^{o}}\left( \frac{%
		n!}{(n-m)!}\right) \left( \frac{o!}{(o-m)!}\right) m^{\lambda ^I+\lambda
		^O}f(n-m,o-m,\lambda ^I,\lambda ^O).
	\end{equation*}%

Recall that the transition probability can be obtained by summing the expression over all possible $(\lambda^I,\lambda^O)$'s:

\begin{equation*}
	p_{n,o;m}:=\sum_{0\leq \lambda ^I\leq n-m,0\leq \lambda ^O\leq
		o-m}	\mathbf{P}(n,o,k^I,k^O ;m,\lambda ^I,\lambda ^O)\text{.}
	\end{equation*}

Hence, we obtain:
	\begin{align*}
	& \sum_{0\leq \lambda ^I\leq n-m}\sum_{0\leq \lambda ^O\leq
		o-m}m^{\lambda ^I}m^{\lambda ^O}f(n-m,o-m,\lambda ^I,\lambda ^O) \\
	=& \sum_{0\leq \lambda ^I\leq n-m}\sum_{0\leq \lambda ^O\leq
		o-m}m^{\lambda ^I}m^{\lambda ^O}{\binom{n-m}{\lambda ^I}}{\binom{o-m}{%
			\lambda ^O}}\times  \\
	& \qquad (o-m)^{n-m-\lambda ^I-1}(n-m)^{o-m-\lambda ^O-1}((n-m)\lambda
	^O+(o-m)\lambda ^I-\lambda ^I\lambda ^O) \\
	=& m\left( \sum_{0\leq \lambda ^I\leq n-m}{\binom{n-m}{\lambda ^I}}%
	m^{\lambda ^I}(o-m)^{n-m-\lambda ^I}\right) \left( \sum_{1\leq \lambda
		^O\leq o-m}{\binom{o-m-1}{\lambda ^O-1}}m^{\lambda
		^O-1}(n-m)^{o-m-\lambda ^O}\right)  \\
	+& m\left( \sum_{1\leq \lambda ^I\leq n-m}{\binom{n-m-1}{\lambda ^I-1}}%
	m^{\lambda ^I-1}(o-m)^{n-m-\lambda ^I}\right) \left( \sum_{0\leq \lambda
		^O\leq o-m}{\binom{o-m}{\lambda ^O}}m^{\lambda ^O}(n-m)^{o-m-\lambda
		^O}\right)  \\
	-& m^{2}\left( \sum_{1\leq \lambda ^I\leq n-m}{\binom{n-m-1}{\lambda ^I-1%
		}}m^{\lambda ^I-1}(o-m)^{n-m-\lambda ^I}\right) \left( \sum_{1\leq
		\lambda ^O\leq o-m}{\binom{o-m-1}{\lambda ^O-1}}m^{\lambda
		^O-1}(n-m)^{o-m-\lambda ^O}\right)  \\
	=& mo^{n-m}n^{o-m-1}+mo^{n-m-1}n^{o-m}-m^{2}o^{n-m-1}n^{o-m-1} \\
	=& mo^{n-m-1}n^{o-m-1}(n+o-m),
	\end{align*}%
	where the first equality follows from  \Cref{forest}, and the third
	follows from the Binomial Theorem.
	
	Multiplying the term $\frac{1}{o^{n}n^{o}}\left( \frac{n!}{(n-m)!}\right)
	\left( \frac{o!}{(o-m)!}\right) $, we get the formula stated in  \Cref{Markov}.

\section{Number of agents matched at each stage of TTC} \label{sec:number match at each stage}

Consider an arbitrary mapping, $g:I\rightarrow O$ and $h:O\rightarrow I$, defined over our finite sets $I$ and $O$. Note that such a mapping naturally  induces a bipartite digraph with vertices $I\cup O$ and directed edges with
the number of outgoing edges equal to the number of vertices, one for each vertex.
In this digraph, $%
i\in I$ points to $g(i)\in O$ while $o\in O$ points to $h(o)\in I$. Such a mapping will be called a bipartite mapping.  A \textbf{cycle} of a bipartite mapping is a cycle in the induced bipartite digraph, namely, vertices $(i_1, o_1, \ldots, i_{k-1}, o_{k-1}, i_k)$ whose nonterminal vertices are distinct and satisfy $g(i_j)=o_j$, $h(o_j)=i_{j+1}$ for $j=1,\ldots,k-1$, and $i_k=i_1$. A
\textbf{random bipartite mapping} selects the pair $(h,g)$ uniformly from $\mathcal{H}\times \mathcal{G}= I^{O}\times O^{I}$; the induced digraph is determined by this pair. Note
that a random bipartite mapping induces a random bipartite digraph
consisting of vertices $I\cup O$ and directed edges emanating from vertices, one for each vertex.
We say that a vertex in  
a digraph is \textbf{cyclic} if it is in a cycle of the digraph.

The following lemma states
the number of cyclic vertices in a random bipartite digraph induced by a
random bipartite mapping.

\begin{lemma}[\citet{jaworski;85}, Corollary 3]\label{jaw} The number $q$ of the cyclic vertices in a random bipartite
	digraph induced by a random bipartite mapping $g:I\rightarrow O$ and $%
	h:O\rightarrow I$ has an expected value of
	\begin{equation*}
	\mathbb{E}[q]:=2\sum_{i=1}^{\min\{o,n\}}\frac{(o)_{i}(n)_{i}}{o^{i}n^{i}},
	\end{equation*}%
	and a variance of
	\begin{equation*}
	8\sum_{i=1}^{\min\{o,n\}}\frac{(o)_{i}(n)_{i}}{o^{i}n^{i}}i-2\mathbb{E}[q]-\mathbb{E}^2[q]
	\end{equation*}
	where $(x)_{j}:=x(x-1)\cdots (x-j+1)$ and $(x)_0:=1$.
\end{lemma}

It is clear that at the beginning of the first round of TTC, if there are
$n$ agents and $o$ objects in the economy, the total number of cyclic
vertices has the same distribution as $q$. Since every cycle contains the
same number of agents and objects, the number $M_{n,o}$ of agents cleared
in that round is $q/2$. Appealing to \Cref{Markov}, the same distributional
statement applies to any round of TTC that begins with $n$ agents and $o$
objects remaining.

Let
\[
        a_i(n,o):=\frac{(o)_i(n)_i}{o^i n^i}.
\]
By Lemma~\ref{jaw},
\[
        \mathbb E[q]
        =
        2\sum_{i=1}^{\min\{o,n\}}a_i(n,o),
\]
and
\[
        \operatorname{Var}(q)
        =
        8\sum_{i=1}^{\min\{o,n\}}i\,a_i(n,o)
        -
        2\mathbb E[q]
        -
        \mathbb E[q]^2.
\]
Therefore,
\[
        \mathbb E[M_{n,o}]
        =
        \sum_{i=1}^{\min\{o,n\}}a_i(n,o),
\]
and
\[
        \operatorname{Var}(M_{n,o})
        =
        2\sum_{i=1}^{\min\{o,n\}}i\,a_i(n,o)
        -
        \mathbb E[M_{n,o}]
        -
        \mathbb E[M_{n,o}]^2.
\]

\citet{jaworski;85} also shows that, as $n$ and $o$ grow,
\[
        \mathbb E[q]
        =
        \sqrt{2\pi\frac{no}{n+o}}\,(1+o(1)),
\]
and
\[
        \operatorname{Var}(q)
        =
        (4-\pi)\frac{2no}{n+o}\,(1+o(1)).
\]
Hence
\[
        \mathbb E[M_{n,o}]
        =
        \sqrt{\frac{\pi no}{2(n+o)}}\,(1+o(1)),
\]
and
\[
        \operatorname{Var}(M_{n,o})
        =
        (4-\pi)\frac{no}{2(n+o)}\,(1+o(1)).
\]
In the balanced case $n=o$, this gives
\[
        \mathbb E[M_{n,n}]
        \sim
        \sqrt{\frac{\pi n}{4}},
        \qquad
        \operatorname{Var}(M_{n,n})
        \sim
        \left(1-\frac{\pi}{4}\right)n.
\]
Equivalently,
\[
        \frac{M_{n,o}}{\sqrt{\pi no/[2(n+o)]}}
\]
has asymptotic mean one and asymptotic variance $(4-\pi)/\pi$.

\section{Number of Rounds Required for TTC and Shapley-Scarf TTC to Conclude} \label{sec:Rounds}

\citet{frieze/pittel;95} analyze Shapley-Scarf TTC.  They obtain a similar Markov chain result for Shapley-Scarf TTC. Our result allows us to compare the two Markov chains. Specifically, we can order the two chains in terms of the likelihood ratio
order. To see this, let us recall the transition probabilities of the
Markov chain obtained by \citet{frieze/pittel;95}:%
\begin{equation*}
\hat{p}_{n;m}=\frac{n!}{n^{m}(n-m)!}\frac{m}{n}
\end{equation*}%
By  \Cref{Markov}, we obtain (assuming $n=o$):%
\begin{eqnarray*}
	p_{n;m} &:&=p_{n,n;m}=\left( \frac{m}{(n)^{2(m+1)}}\right) \left( \frac{n!}{%
		(n-m)!}\right) ^{2}(2n-m) \\
	&=&\left( \frac{n!}{n^{m}(n-m)!}\right) ^{2}\left( \frac{m(2n-m)}{n^{2}}%
	\right) .
\end{eqnarray*}

Let us compare the two distributions in terms of likelihood ratio order. Fix
$n\geq 1$ and any $m^{\prime }\geq m$. It is easy to check that%
\begin{equation*}
\frac{\hat{p}_{n,m^{\prime }}}{\hat{p}_{n,m}}=\frac{n^{m}(n-m)!}{%
	n^{m^{\prime }}(n-m^{\prime })!}\frac{m^{\prime }}{m}
\end{equation*}%
while%
\begin{equation*}
\frac{p_{n,m^{\prime }}}{p_{n,m}}=\left( \frac{n^{m}(n-m)!}{n^{m^{\prime
		}}(n-m^{\prime })!}\right) ^{2}\left(\frac{m^{\prime }}{m}\right)\left(\frac{2n-m^{\prime }}{%
		2n-m}\right)\text{.}
	\end{equation*}%
	Now, observe that%
	\begin{eqnarray*}
		\left( \frac{\hat{p}_{n,m'}}{\hat{p}_{n,m}}\right)^{-1}
		\frac{
			p_{n,m'}}{p_{n,m}} &=&\left(\frac{1}{n^{m'-m}}\right)
		\left( \frac{(n-m)!}{(n-m')!}\right) \frac{(2n-m^{\prime })}{(2n-m)}
		\\
		&=&\frac{(n-m)(n-m-1)...(n-m^{\prime }+1)}{n^{m^{\prime }-m}}\left(\frac{2n-m^{\prime }}{2n-m}\right)\leq 1.
	\end{eqnarray*}%
	This proves that for any $n$, the distribution $\hat{p}_{n,\cdot }$
	dominates $p_{n,\cdot }$ in terms of likelihood ratio order. {This means that, given that $n$ agents/objects are remaining at a given round, for any $m$, the probability that  $m$ or fewer agents are part of a cycle is smaller under Shapley-Scarf TTC  than under TTC. That is, more agents and objects are likely to be cleared under Shapley-Scarf TTC than under TTC in each round,  starting from the same number of agents/objects. This implies that TTC tends to last longer  than the Shapley-Scarf TTC: for each $t \geq 1$, the probability that TTC stops before Round $t$ is smaller than the probability that Shapley-Scarf TTC stops before Round $t$.  More precisely, the random round at which TTC stops first-order stochastically dominates the random round at which the Shapley-Scarf TTC stops.}

\section{Simulation Protocols and Robustness Figures}
\label{oa:simulation-details}

This section records the simulation protocol behind \Cref{fig}, \Cref{fig2}, \Cref{fig3}, and \Cref{fig_bis}.
  In each simulated market, we draw the
relevant preference and priority primitives, run TTC and RSD on the same
realized market, and compute the relevant justified-envy statistics.

For the baseline one-to-one simulations, each agent's preference ordering
over objects and each object's priority ordering over agents are drawn
independently and uniformly. For RSD, an additional serial order is drawn
uniformly. For each market size $n$, the reported statistic is averaged over
100 independent market draws.

For the correlated-preference simulations, utilities take the form
\[
        u_i(o)=u_o+\xi_{io},
\]
where $u_o$ and $\xi_{io}$ are drawn independently from $U[0,1]$.
Preferences are induced by the resulting utility ranking. Object priorities
are generated from independent priority scores.

For the many-to-one simulations, the set of object types is fixed as
specified in the figure note, and each type has the stated capacity. We implement TTC and RSD in the usual way (e.g., \cite{abdulkadiroglu/sonmez:03}) in the induced seat-level economy. The reported
statistics are then aggregated over agents and object types according to the
definitions in the main text.

The main simulation metric in  Panel (a) of \Cref{fig}, \Cref{fig2}, \Cref{fig3}, and Panel (a) of \Cref{fig_bis} is
\[
\frac{
\sum_{i,j}\mathbf 1\{i\text{ justifiably envies }j\}
}{
\sum_{i,j}\mathbf 1\{i\text{ envies }j\}
}.
\]
Panel (b) of  \Cref{fig_bis} instead reports the fraction of agents who
experience justified envy among agents who experience envy.




\end{document}